\documentclass[11pt]{article}
\usepackage[letterpaper, left=1in, right=1in, top=1in,bottom=1in]{geometry}
\PassOptionsToPackage{linesnumbered,noend}{algorithm2e} 
\PassOptionsToPackage{pagebackref=true}{hyperref} 

\makeatletter
\renewcommand{\@biblabel}[1]{[#1]\hfill}
\makeatother

\usepackage[utf8]{inputenc}

\usepackage{algorithm,comment}
\usepackage[noend]{algpseudocode}
\usepackage{auxiliary}

\usepackage[utf8]{inputenc}

\usepackage[numbers]{natbib}

\usepackage[noend]{algpseudocode}

\usepackage[utf8]{inputenc} 
\usepackage[T1]{fontenc}    
\usepackage[pagebackref=true]{hyperref}       
\usepackage{url}            
\usepackage{amsfonts}       
\usepackage{nicefrac}       
\usepackage{microtype}      
\usepackage{bbm}
\usepackage{upgreek}
\usepackage{verbatim}

\newcommand{\cov}{\mathrm{Cov}}
\newcommand{\rew}{\mathrm{Rew}}
\newcommand{\ZZ}{\mathbb{Z}}
\newcommand{\twe}{\mathrm{TWErr}}
\newcommand{\loce}{\mathrm{LocErr}}
\newcommand{\rewe}{\mathrm{RewErr}}
\newcommand{\twem}{\tau_M}
\newcommand{\locem}{\Lambda_M}
\newcommand{\rewem}{\rho_M}
\newcommand{\lmin}{L_{\min}}
\newcommand{\lmax}{L_{\max}}
\newcommand{\cred}{\mathrm{Cr}}

\newcommand{\remove}[1]{}

\begin{document}


\title{Online Time-Windows TSP with Predictions}

 \author{
  Shuchi Chawla\thanks{University of Texas-Austin, \texttt{shuchi@cs.utexas.edu}} 
  \and 
    Dimitris Christou\thanks{University of Texas-Austin, \texttt{christou@cs.utexas.edu}}  
 }

\date{\today}

\maketitle


\begin{abstract}

In the {\em Time-Windows TSP} (TW-TSP) we are given requests at different locations on a network; each request is endowed with a reward and an interval of time; the goal is to find a tour that visits as much reward as possible during the corresponding time window. For the online version of this problem, where each request is revealed at the start of its time window, no finite competitive ratio can be obtained. We consider a version of the problem where the algorithm is presented with predictions of where and when the online requests will appear, without any knowledge of the quality of this side information. 

Vehicle routing problems such as the TW-TSP can be very sensitive to errors or changes in the input due to the hard time-window constraints, and it is unclear whether imperfect predictions can be used to obtain a finite competitive ratio. We show that good performance can be achieved by explicitly building slack into the solution. Our main result is an online algorithm that achieves a competitive ratio logarithmic in the diameter of the underlying network, matching the performance of the best offline algorithm to within factors that depend on the quality of the provided predictions. The competitive ratio degrades smoothly as a function of the quality and we show that this dependence is tight within constant factors. 
\end{abstract}

\addtocounter{page}{-1}
\thispagestyle{empty}


\newpage

\section{Introduction}\label{sec:introduction}


Many optimization problems exhibit a large gap in how well they can be optimized offline versus when their input arrives in online fashion. In order to obtain meaningful algorithmic results in the online setting, a natural direction of investigation is to consider "beyond worst case" models that either limit the power of the adversary or increase the power of the algorithm. A recent line of work in this direction has considered the use of {\em predictions} in bridging the offline versus online gap. Predictions in this context are simply side information about the input that an online algorithm can use for its decision making; the true input is still adversarially chosen and arrives online. The goal is to show that on the one hand, if the predictions are aligned with the input, the algorithm performs nearly as well as in the offline setting (a property known as {\em consistency}); and on the other hand, if the predictions are completely unrelated to the input, the algorithm nevertheless performs nearly as well as the best online algorithm (a.k.a. {\em robustness}). {\em Put simply, good predictions should help, but bad predictions should not hurt, and ideally we should reap the benefits without any upfront knowledge about the quality of the predictions.}

Predictions have been shown to effectively bypass lower bounds for a variety of different online decision-making problems including, for example, caching~\cite{LV21, R20, JPS20}, scheduling~\cite{BFMMN21,IKQP21,ALT21}, online graph algorithms~\cite{APT21}, load balancing~\cite{LLMV20,LMRX21}, online set cover~\cite{BMS20}, matching problems~\cite{DILMV21},  $k$-means~\cite{EFSWZ22}, secretary problems~\cite{AGKK20, DLLV21}, network design~\cite{ErlebachLMS22, XM22, BBFL23} and more.\footnote{A comprehensive compendium of literature on the topic can be found at \cite{ALPS-web}.}  


In this paper, we consider a problem whose objective function value is highly sensitive to changes in the input, presenting a novel challenge for the predictions setting. In the Traveling Salesman Problem with Time Windows (TW-TSP for short), we are given a sequence of service requests at different locations on a weighted undirected graph. Each request is endowed with a reward as well as a time window within which it should be serviced. The goal of the algorithm is to produce a path that maximizes the total reward of the requests visited within their respective time windows. In the online setting, the requests arrive one at a time at the start of their respective time windows, and the algorithm must construct a path incrementally without knowing the locations or time windows of future requests.

Vehicle routing problems such as the TW-TSP that involve hard constraints on the lengths of subpaths (e.g. the time at which a location is visited) are generally more challenging than their length-minimization counterparts. In particular, a small bad decision at the beginning of the algorithm, such as taking a slightly suboptimal path to the first request, can completely obliterate the performance of the algorithm by forcing it to miss out on all future reward. In the offline setting, this means that the approximation algorithm has to carefully counter any routing inefficiency in some segments by intentionally skipping reachable value in other segments. In the online setting, this means that no sublinear competitive ratio is possible. 



\vspace{0.05in}
\begin{center}
\parbox{0.95\textwidth}{
{\em Given the sensitivity of the TW-TSP objective to small routing inefficiencies, is it possible to design meaningful online algorithms for this problem using imperfect predictions?} }
\end{center}
\vspace{0.05in}

We consider a model where the algorithm is provided with a predicted sequence of requests at the beginning, each equipped with a predicted location and a predicted time window. The true sequence of requests is revealed over time as before. Of course if the predicted sequence is identical to the true request sequence, the algorithm can match the performance of the best offline algorithm. But what if the predictions are slightly off? Could these small errors cause large losses for the online algorithm? Can the algorithm tolerate large deviations?

Our main result is an online algorithm for the TW-TSP based on predictions whose performance degrades smoothly as a function of the errors in prediction. We obtain this result by explicitly building slack into our solution and benchmark. In a slight departure from previous work on TW-TSP, we require the server to spend one unit of idle "service time" at each served request. We show that this is necessary to obtain a sublinear approximation even with predictions (Theorem~\ref{t:lb_for_0_sc}). (However, in the absence of predictions, the setting with service times continues to admit a linear lower bound on the competitive ratio; See Theorem~\ref{l:unbounded_rand}.) We then use service times judiciously in planning a route and accounting for delays caused by prediction errors.

There are three sources of error in predictions: (1) the predicted locations of requests may be far from the true locations; (2) the predicted time windows may be different from the true time windows; (3) the predicted rewards may be different from the true rewards. We show a polynomial dependence of the competitive ratio of our algorithm on the first and third components, and prove that this dependence is tight to within constant factors. Furthermore, we show that as long as the deviation in time windows is no more than half the minimum time window length, this error causes at most a further constant factor loss in performance. Besides this dependence on the prediction error, the competitive ratio of our algorithm depends logarithmically on the diameter of the underlying network, matching the performance of the best known offline algorithm for TW-TSP.

Importantly, our algorithm requires little to no information about how the predictions match up against the true requests. For the purpose of analysis, we measure the error in predictions with respect to some underlying matching between the predicted and true requests -- the error parameters are then defined in terms of the maximum mismatch between any predicted request and its matched true request. This matching is never revealed to the algorithm and in fact the performance of the algorithm depends on the quality of the {\em best} possible matching between the predicted and true requests. The only information the algorithm requires about the quality of the predictions is the location error -- the maximum distance between any prediction and its matched true request. Even for this parameter, an upper bound suffices (and at a small further loss, a guess suffices).

Our overall approach has several components.  The first of these is to construct an instance of the TW-TSP over predicted requests that requires the server to spend some idle time at each request as a "service delay". We then extend offline TW-TSP algorithms to this service delay setting, obtaining a logarithmic in diameter approximation. We then follow and adapt this offline solution in the online setting. Every time the offline solution services a predicted request, we match this request to a previously revealed true request, take a detour from the computed path to visit and service the true request, and then resume the precomputed path. Altogether this provides the desired competitive ratio.

Our results further generalize to a setting where predictions are coarse in that each single predicted location captures multiple potential true requests that are nearby. We show that with prediction errors defined appropriately, we can again achieve a competitive ratio for this "many to one matching" setting that is logarithmic in the diameter of the graph and polynomial in the prediction error. 


\subsection*{Further related work}

Using predictions in the context of online algorithm design was first proposed by \cite{LV21} for the well-studied caching problem. Since that work, the literature on online algorithm design with predictions has grown rapidly. We point the interested reader to a compendium at \cite{ALPS-web} for further references.

\paragraph{Metric error with outliers.} Azar et al.~\cite{APT21} initiated the study of predictions in the context of online graph optimization problems, and proposed a framework for quantifying errors in predictions, called {\em metric error with outliers}, that we adapt. The idea behind this framework is to capture two sources of error: (1) Some true requests may not be captured by predictions and some predictions may not correspond to any true requests; (2) For requests that are captured by predictions, the predictions may not be fully faithful or accurate. The key observation is that it is possible to design algorithms with performance that depends on these two sources of error {\em without
explicit upfront knowledge of the (partial) correspondence between predicted and true requests}. 

In this work, we focus mostly on the second source of error, which we further subdivided into three kinds of error in order to obtain a finer understanding of the relationship between the competitive ratio and different kinds of error. As in the work of Azar et al.~\cite{APT21}, we assume that the correspondence between predicted and true requests is never explicitly revealed to the algorithm. The performance of the algorithm nevertheless depends on the error of the {\em best} matching between predicted and true requests. In Section~\ref{sec:approach} we describe how the first source of error in Azar et al.'s framework can also be incorporated into our bounds.



\paragraph{TSP with predictions.} Recently a few papers~\cite{bernardini2022universal,HWLCL22, gouleakis2022learningaugmented} have considered the online TSP and related routing problems with predictions. The input to the online TSP is similar to ours: requests arrive over time in a graph, and a tour must visit each request after its arrival time. However, the objective is different. In our setting, requests also have deadlines, and the algorithm cannot necessarily visit all requests. The goal therefore is to visit as many as possible. In the online TSP, there are no deadlines, and so the objective is to visit {\em all} requests as quickly as possible, or in other words to minimize the makespan. This makespan minimization objective is typically much easier than the deadline setting, as evidenced by constant factor approximations for it in the offline, online, and predictions settings, as opposed to logarithmic or worse approximations for the latter problem.

\paragraph{TW-TSP {\em without} predictions.} The (offline) TW-TSP problem has a rich literature and has been studied for over $20$ years. The problem is known to be NP-hard even for special cases, e.g. on the line~\cite{T92}, and when all requests have the same release times and deadlines (a.k.a. Orienteering)~\cite{BCKLM03}. Orienteering admits constant factor approximations~\cite{BCKLM03, BBCM04, CKP12}, and even a PTAS when requests lie in a fixed dimensional Euclidean space~\cite{AMN98, CHP06}. For general time windows, constant factor approximations are only known for certain special cases: e.g. constant number of distinct time windows~\cite{CK04}; and on line graphs~\cite{T92,KN03,BYES05,GJJZ20}. For general graphs and time-windows, the best approximations known are logarithmic in input parameters~\cite{BBCM04,CKP12}. 


To the best of our knowledge, the online setting for TW-TSP has only been considered by Azar and Vardi~\cite{azar15}. Azar and Vardi assume that service times are non-zero and present competitive algorithms under the assumption that the smallest time window length $\lmin$ is comparable to the diameter $D$ of the graph; No sublinear competitive ratio can be achieved if $\lmin<D/2$. We are able to beat this lower bound by relying on predictions.

\subsection*{Organization of the paper}

We formally define the problem and our error model in Section~\ref{sec:prelim}. Section~\ref{sec:approach} describes our results and provides an outline for our analysis. All of our main results are covered in that section. Proofs of these results can be found in subsequent sections. Sections~\ref{sec:relating_opt}, \ref{sec:offline}, and \ref{sec:online} fill in the details of our upper bound, Section~\ref{sec:lower_bound} proves lower bounds, and Section~\ref{sec:extension_proofs} analyzes the extension of our main result to the "many to one matching" setting described above. Proofs skipped in the main body of the paper can be found in the appendix.

\section{Definitions}\label{sec:prelim}

\subsection{The Traveling Salesman Problem with Time-Windows}


An instance of the TW-TSP consists of a network $G$ and a (finite) sequence of service requests $I$. $G=(V,E, \ell)$ is an undirected network with edge lengths $\{\ell_e\}_{e\in E}$. Extending the notion of distance to all vertex pairs in $G$, we define $\ell(u,v)$ for $u,v\in V$ to be the length of the shortest path from $u$ to $v$. We assume without loss of generality that $G$ is connected and that the edge lengths $\ell_e$ are integers. A service request $\sigma = (v_\sigma, r_\sigma, d_\sigma, \pi_\sigma)$ consists of a vertex $v_\sigma\in V$ at which the request arrives, a release time $r_\sigma\in\ZZ^+$, a deadline $d_\sigma\in\ZZ^+$ with $d_\sigma > r_\sigma$, and a reward $\pi_\sigma\in\ZZ^+$. We use $\Sigma \subseteq V\times \ZZ^+\times\ZZ^+\times\ZZ^+$ to denote the set of all possible client requests and $I\subset \Sigma$ to denote the set of requests received by the algorithm.


The solution to TW-TSP is a continuous {\em walk} on $G$ that is allowed to remain idle on the vertices of the graph.\footnote{To keep our exposition simple, we do not specify a starting location for the walk. However, all of our algorithms can be adapted without loss to the case where a starting location is fixed, as described towards the end of Section~\ref{sec:approach}.}  Formally, the walk starts from some vertex at time $t=0$; at every time-step that it occupies a vertex $u\in V$, it can either remain idle on $u$ for some number of time-steps \textit{or} it can move to some $v\in V$ by spending time $t=\ell_{(u,v)}$; we comment that while the path is mid-transition, no more decisions can be made. Notice that this creates a notion of a discrete time-horizon that will be important towards formalizing the online variant of the problem. 

We use $\mathcal{W}(G)$ to denote the set of all walks on $G$. Given a request $\sigma \in \Sigma$, we say that a walk $W$ \textit{covers} it if $W$ remains idle on vertex $v_\sigma$ for at least one time-step,\footnote{As we mentioned in the introduction, this requirement of a minimal one-unit service time is necessary in order to achieve any sublinear approximation for the online TW-TSP even with predictions. See Theorem~\ref{t:lb_for_0_sc} in Section~\ref{sec:lower_bound}.} starting on some step $\tau \in [r_\sigma,d_\sigma-1]$. 
For a sequence of requests $I\subset \Sigma$, we use $\cov(W,I)\subseteq I$ to denote the set of requests in $I$ that are covered by $W$. Then, the reward obtained by $W$ is denoted by 
$\rew(W,I) := \sum_{\sigma \in \cov(W,I)}\pi_\sigma$. The objective of TW-TSP is to compute a walk $W\in\mathcal{W}(G)$ of maximum reward. We denote this by $\opt(G, I)$:
\[\opt(G,I) : = \max_{W\in\mathcal{W}(G)}\left[\rew(W,I)\right]\]

\subsection{The offline, online, and predictions settings}

We assume that the network $G$ is known to the algorithm upfront in all of the settings we consider. In the {\bf offline} version of the problem, the sequence of requests $I$ is given to the algorithm in advance. In the {\bf online} version, requests $\sigma\in I$ arrive in an online fashion; specifically, each request $\sigma\in I$ is revealed to the algorithm at its release time $r_\sigma$. 

In the {\bf predictions} setting, the true sequence of requests $I$ arrives online, as in the online setting. However, the algorithm is also provided with a predicted sequence $I'\subset \Sigma$ in advance, where every request $\sigma'\in I'$ is endowed with a location, a time window, and a reward. The quality of predictions is expressed in terms of their closeness to true requests. To this end, we define three notions of mismatch or error. For a true request $\sigma$ and predicted request $\sigma'$, the location error, time windows error, and reward error are defined as:
\begin{align*}
    \loce(\sigma,\sigma') & :=\ell(v_\sigma,v_{\sigma'})\\
    \twe(\sigma,\sigma') & :=\max\{|r_\sigma-r_{\sigma'}|, |d_\sigma-d_{\sigma'}|\}\\
    \rewe(\sigma, \sigma') & :=\max\{\pi_\sigma/\pi_{\sigma'}, \pi_{\sigma'}/\pi_\sigma\}
\end{align*}





We extend these definitions to the entire sequences $I$ and $I'$ through an underlying (but unknown to the algorithm) matching between the requests in the two lists:

\begin{definition}
Given two request sequences $I,I'\subset\Sigma$ with $|I|=|I'|$ and a perfect matching $M:I\mapsto I'$, we define the location, time window, and reward errors for the matching $M$ as:
\begin{align*}
\locem  :=\max_{\sigma\in I} \loce(\sigma, M(\sigma)), \quad \twem  :=\max_{\sigma\in I} \twe(\sigma, M(\sigma)), \text{ and } \quad 
\rewem  :=\max_{\sigma\in I} \rewe(\sigma, M(\sigma))
\end{align*}


\end{definition}

We use $n=|V|$ to denote the number of vertices in $G$, $D$ to denote the diameter of the graph, and $\lmin$ and $\lmax$ to denote the size of the smallest and largest time windows respectively (of a true or predicted request) in the given instance: $\lmin=\min_{\sigma\in I\cup I'} |d_\sigma-r_\sigma|$ and $\lmax=\max_{\sigma\in I\cup I'} |d_\sigma-r_\sigma|$. The competitive ratios of the algorithms we develop depend on these parameters.

\paragraph{Knowns and unknowns.} We denote an instance of the TW-TSP with predictions by $(G, I, I', M)$. All components of the instance are chosen adversarially. As mentioned earlier, the network $G$ and the predicted sequence $I'$ are provided to the algorithm at the start. The sequence $I$ arrives online. We assume that the algorithm receives no direct information about the matching $M$, but is provided with an upper bound on the error $\locem$. We will also assume that the algorithm knows the parameter $\lmin$, although this is without loss of generality as the parameter can be inferred within constant factor accuracy from the predictions.\footnote{In particular, assuming $\twem\le\lmin/2$, which is necessary for our results to hold, the time window of any true request can be no shorter than half the smallest time window of any predicted request.}

\subsection{The TW-TSP with service times}

At a high level our algorithm has two components: an offline component that computes a high-reward walk over the predicted locations of requests, and an online component that largely follows this walk but takes "detours" to cover the arriving true sequence of requests. In particular, as the algorithm follows the offline walk, for each predicted location it visits where a "close by" true request is available, the algorithm takes a "detour" to this true request, returns back to the predicted location, and resumes the remainder of the walk. In order to incorporate the time spent taking these detours in our computation of the offline walk, we require the walk to spend some "service time" at each predicted location it covers. Accordingly, we define a generalization of the TW-TSP:

\begin{definition}
The {\em TW-TSP with Service Times} (TW-TSP-S) takes as input a network $G$, a sequence of service requests $I$, and a service time $S\in\ZZ^+$, and returns a walk $W\in \mathcal{W}(G)$. We say that $W$ covers a request $\sigma\in I$, denoted $\sigma\in \cov(W,I,S)$, if it remains idle on vertex $v_\sigma$ for at least $S$ time steps, starting at some step $t\in [r_\sigma, d_\sigma-S]$. We define the reward of $W$ as $\rew(W,I,S) :=\sum_{\sigma\in\cov(W, I, S)} \pi_\sigma$. The optimal value of the instance is given by:
\[\opt(G, I, S) := \max_{W\in \mathcal{W}(G)}[\rew(W,I,S)]\]
\end{definition}

Note that the original version of TW-TSP as defined previously simply corresponds to the special case of TW-TSP-S with service time $S=1$, and in particular, we have $\rew(W,I)=\rew(W,I,1)$, and $\opt(G,I)=\opt(G,I,1)$.

\section{Our results and an outline of our approach}\label{sec:approach}
Our main result is as follows.

\begin{theorem}\label{t:main}
Given any instance $(G, I, I', M)$ of the TW-TSP with predictions whose errors satisfy $\twem\le\lmin/2$ and $\locem\le (\lmin-1)/4$, there exists a polynomial-time online algorithm that takes the tuple $(G, I', \locem)$ as offline input and $I$ as online input, and constructs a walk $W\in \mathcal{W}(G)$ such that
\[\mathbb{E}[\rew(W,I)]\ge \frac{1}{O(\locem\cdot\rewem^2\cdot\log \min(D,\lmax))}\cdot \opt(G,I)\]
\end{theorem}


As mentioned previously, our algorithm consists of two components. The offline component constructs a potential walk in the network with the help of the predicted requests. Then an online component adapts this walk to cover true requests that arrive one at a time. We break up the design and analysis of our algorithm into four steps. The first two steps relate the offline instance we solve to the hindsight optimal solution for the online instance. The third step then applies an offline approximation to the predicted instance with appropriate service times. The final step deals with the online adaptation of the walk to the arriving requests.

The following four lemmas capture the four steps. First, we show (Section~\ref{sec:relating_opt}) that introducing a service time of $S$ hurts the optimal value by at most a factor of $2S-1$. As we prove in Lemma~\ref{l:service-costs-lb} of Section~\ref{sec:relating_opt}, this dependency on $S$ is tight. Observe that we require $S\le\lmin$, as for any tour to feasibly cover a request, the service time for that request must fit within its time window.

\begin{lemma}
\label{l:service-costs}
    For any instance $(G,I)$ of the TW-TSP with service times, and any integer $S\le\lmin$, we have \[\opt(G,I,S)\ge\frac{1}{2S-1}\cdot \opt(G,I,1).\]
\end{lemma}

Our second step (also in Section~\ref{sec:relating_opt}) relates the value of the optimal solution over the true requests $I$ to the optimum over the predicted sequence $I'$. In both cases, we impose some service time requirements. Note that this argument needs to account for the discrepancy in locations, time windows, as well as the rewards of the true and predicted requests. 

\begin{lemma}
\label{l:relating-opts}
    Let $(G,I,I',M)$ be an instance of the TW-TSP with predictions, where $\locem$, $\rewem$, and $\twem$ denote the maximum location, reward, and time window errors of the instance respectively. Define $S:=4\locem+1$ and $S':=2\locem+1$. Then, if $\twem \leq \lmin/2$ and $\locem\leq (\lmin-1)/4$, we have \[\opt(G,I',S')\ge\frac{1}{3\rewem}\cdot \opt(G,I,S).\]
\end{lemma}

Our third step (Section~\ref{sec:offline}) captures the offline component of our algorithm: computing an approximately optimal walk over the predicted requests with the specified service times. For this we leverage previous work on the TW-TSP without service times and show how to adapt it to capture the service time requirement.

\begin{lemma}
\label{l:offline-alg}
    Given any instance $(G,I',S')$ of the TW-TSP with service times, there exists a polynomial time algorithm that returns a walk $W\in\mathcal{W}(G)$ with reward 
    \[\rew(W,I',S')= \frac{1}{\mathcal{O}\left(\log \min(D,\lmax)\right)} \cdot \opt(G,I',S').\]

\end{lemma}

Finally, the fourth component (Section~\ref{sec:online}) addresses the online part of our algorithm. Given a walk computed over the predicted request sequence, it solves an appropriate online matching problem to construct detours to capture true requests. As in Lemma~\ref{l:relating-opts}, this part again needs to account for the discrepancy in locations, time windows, as well as the rewards of the true and predicted requests.

\begin{lemma}
\label{l:online-alg}
    Given an instance $(G,I,I',M)$ of the TW-TSP with predictions satisfying $\twem\leq\lmin/2$ and $\locem\leq (\lmin-1)/4$; a walk $W'\in \mathcal{W}(G)$; and any integer $S'\geq 2\locem + 1$, there exists an online algorithm (Algorithm~\ref{alg:matching_alg}) that returns a walk $W\in\mathcal{W}(G)$ with expected reward
    \[\mathbb{E} \left[ \rew(W,I,1) \right] \ge \frac{1}{6\rewem}\cdot \rew(W',I',S').\]
\end{lemma}

\noindent
Theorem~\ref{t:main} follows immediately by putting Lemmas~\ref{l:service-costs}, \ref{l:relating-opts}, \ref{l:offline-alg} and \ref{l:online-alg} together.

\subsection*{Lower bounds and tightness of our results.} 

We show that the online TW-TSP does not admit sublinear competitive algorithms in the absence of predictions if $\lmin<D$, even with non-zero service times. Furthermore, if the service times are all $0$, no sublinear competitive ratio is possible even using predictions that are accurate in all respects except the request location. Therefore, in order to achieve a nontrivial competitive ratio, it is necessary to use predictions as well as to impose non-zero service times on the optimum. The proofs are presented in Section~\ref{sec:lower_bound}.

\begin{theorem}\label{l:unbounded_rand}
The competitive ratio of any randomized online algorithm for Online TW-TSP on instances with $\lmin \leq D$ and all service times equal to $1$ is at most $1/n$.
\end{theorem}

\begin{theorem}\label{t:lb_for_0_sc}
    For any $S>0$, there exists an instance $(G,I,I',M)$ of the TW-TSP with predictions and service times $0$, satisfying $\twem=0$, $\rewem=1$, and $\locem=S$, such that any randomized online algorithm taking the tuple $(G, I', \locem)$ as offline input and $I$ as online input achieves a reward no larger than $O(1/n)\cdot\opt(G,I,0)$. Here $n$ is the number of vertices in $G$.
\end{theorem}

As mentioned earlier, the best known approximation factor for the offline TW-TSP is $O(\log \lmax)$ (which we show can be improved slightly to $O(\log \min(D,\lmax))$). We inherit this logarithmic dependence on $D$ and $\lmax$ in the predictions setting. Furthermore, any improvements to the offline approximation would immediately carry through into our competitive ratio as well. In particular, given an offline TW-TSP algorithm that achieves a competitive ratio of $\alpha(D,\lmax)$, we obtain an online algorithm that achieves a competitive ratio of $O(\locem\cdot\rewem^2\cdot\alpha(D,\lmax))$. 

The dependence of our bound on $\rewem$ can easily be seen to be tight -- consider a star graph with requests on leaves, and edge lengths and time windows defined in such a manner that any feasible walk can cover at most one request. Then an uncertainty of a factor of $\rewem$ in the predicted rewards can force any online algorithm to obtain an $\Omega(\rewem^2)$ competitive ratio even if the predictions are otherwise perfect. Finally, we show in Section~\ref{s:loc-err-lb} that the dependence of our competitive ratio on $\locem$ is also tight:
\begin{theorem}\label{t:lb}
    For any $S>0$, there exists an instance $(G, I, I', M)$ of the TW-TSP with predictions satisfying $\twem=0$, $\rewem=1$, and $\locem=S$ such that the competitive ratio of any randomized online algorithm taking the tuple $(G, I', \locem)$ as offline input and $I$ as online input asymptotically approaches $1/(S+1)$.
\end{theorem}

\subsection*{Extensions and generalizations.} 

We now describe some ways in which we can weaken the assumptions in Theorem~\ref{t:main} while maintaining its competitive ratio guarantee:
\begin{itemize}
    \item {\bf Lack of knowledge of $\locem$.} Our algorithm continues to work as intended if it is provided with an upper bound on $\locem$ rather than the exact value of the parameter, with the performance of the algorithm degrading linearly with the upper bound, as in the theorem above. One such upper bound is simply $\lmin/4$. Moreover, by guessing $\locem$ within a factor of $2$ in the range $[0,\lmin/4]$, we can obtain the claimed approximation with a further loss of $\mathcal{O}(\log \lmin)$. Thus, our algorithm can achieve non-trivial guarantees that scale with the location error even in settings where no information is given about any of the predicton errors $\locem, \twem,\rewem$.

    \item {\bf Assumptions on $\twem$ and $\locem$.} It is easy to see that it is necessary to assume $\locem\le\lmin$ to obtain a nontrivial competitive ratio, as predictions with a location error larger than the time window size are of no value to the online algorithm. On the other hand, assuming $\twem\le\lmin$ is not necessary. We can accommodate larger time window errors by following one out of roughly $\twem/\lmin$ different time shifts of the offline walk. This worsens our approximation factor by an additional factor of $\twem/\lmin$.
    


    \item {\bf Random rewards.} Our results also hold in the case of random rewards. Specifically, consider a setting where the rewards $\{\pi_\sigma\}_{\sigma\in I}$ are drawn from some joint (not necessarily product) distribution $D$ over $\mathbb{R}_+^{I}$. 
    In that case, we define $\rew(W,I):=\sum_{\sigma\in \cov(W,I)}\operatorname{E}[\pi_\sigma]$, and $\opt(G,I)$ as the maximum reward obtained by any walk $W\in\mathcal{W}(G)$.\footnote{Note that we do not allow the optimal walk to adapt to instantiations of rewards. Adaptive walks cannot be competed against in an online setting even with predictions.}
    Finally, we define $\rewe(\sigma,\sigma')$ as the mismatch between $\pi_\sigma'$ and $\operatorname{E}[\pi_\sigma]$.
    Our analysis provides the same approximation as before in this setting. See Appendix~\ref{s:app-ext-proofs} for a formal proof.

    \begin{corollary}\label{c:random-rewards}
    Given an instance $(G, I, I', M)$ of the TW-TSP with predictions where requests have randomly drawn rewards, and predictions errors satisfy that $\twem\le\lmin/2$ and also $\locem\le (\lmin-1)/4$, there exists a polynomial-time online algorithm that takes the tuple $(G, I', \locem)$ as offline input and $I$ as online input, and constructs a walk $W\in \mathcal{W}(G)$ such that
    \[\mathbb{E}[\rew(W,I)]\ge \frac{1}{O(\locem\cdot\rewem^2\cdot\log \min(D,\lmax))}\cdot \opt(G,I)\]
    \end{corollary}

    \item {\bf Rooted instances.} Next, we consider the case where a starting vertex $v_0$ is also specified, and the solution space $\mathcal{W}(G)$ includes all walks on $G$ that \textit{start} on vertex $v_0$ at $t=0$. We can easily see that this setting is essentially equivalent to its unrooted counterpart, under the extra assumption that each request $\sigma=(v_\sigma,r_\sigma,d_\sigma,\pi_\sigma)$ satisfies the conditions $\ell(v_0,v_\sigma)\leq r_\sigma$. This is a reasonable assumption as no algorithm can visit a request $\sigma$ before time $\ell(v_0,v_\sigma)$ anyway. Clearly, for any rooted instance $(G,I,v_0)$, the unrooted optimal $\opt(G,I)$ is an upper bound on the rooted optimal $\opt(G,I,v_0)$. On the other hand, the unrooted path computed by our algorithm can be transformed to a path of same reward rooted at $v_0$ by going directly from $v_0$ to the predicted request serviced first, as this distance is at most equal to the request's release time.

    \item {\bf Partial matching.} Next we consider the case where not all true requests are captured by the predicted requests and, on the flip side, where some predicted requests do not correspond to true requests at all. Following the framework of \cite{APT21}, we consider partial matchings between $I$ and $I'$, and define $\Delta_1^M$ to be the total reward of all true requests that are unmatched, and $\Delta_2^M$ to be the total predicted reward of predicted requests that are unmatched. Then, it is easy to see that our analysis goes through for the subsets of $I$ and $I'$ that are matched to each other, costing us an additive amount of no more than $\Delta_1^M+\Delta_2^M$. See Appendix~\ref{s:app-ext-proofs} for a formal proof.

    \begin{corollary}\label{cor:incomplete_matching}
    Given an instance $(G,I,I')$ of the TW-TSP with predictions, let $M$ be any (incomplete) matching between $I$ and $I'$, and let the error parameters $\locem, \rewem, \twem, \Delta_1^M,$ and $\Delta_2^M$ be defined as above. Then, there exists an online algorithm that takes $(G, I', \locem)$ as offline input and $I$ as online input, and returns a walk $W\in\mathcal{W}(G)$ such that
     \[\mathbb{E}\left[\rew(W,I)\right] \geq \Omega\left(\frac{1}{\locem\cdot \rewem^2\cdot \log \min(D,\lmax)}\right)\cdot \left(\opt(G,I) - \Delta_1^M\right) - \frac{\Delta_2^M}{\rewem}.\]
    \end{corollary}

    \item {\bf Many to one matching.} Consider a setting where predictions are coarse in that each single predicted location captures multiple potential true requests. We can model such a setting within our predictions framework and obtain almost the same guarantee as in Theorem~\ref{t:main}. In particular, for this setting, let $M$ be a many-to-one matching from $I$ to $I'$. We define the location error of a predicted request $\sigma'\in I'$ as the length of the shortest path that starts at $\sigma'$, visits all of the locations of the true requests that are preimages of $\sigma'$ in $M$, spending one unit of time at each, and returns back to $\sigma'$. Observe that this location error is the length of the optimal solution to an orienteering problem rooted at $\sigma'$. Correspondingly, we want the reward associated with $\sigma'$ to capture the total reward of all the true requests matched to $\sigma'$, and define its reward error accordingly. Finally, the time window error is defined as before, as a maximum over all pairs $\sigma$ and $\sigma'$ that are matched to each other. Our algorithm for the setting of Theorem~\ref{t:main} constructs a matching between $I'$ and $I$ in an online fashion. For this one to many setting, we solve instances of the orienteering problem rooted at each predicted request we visit. The performance of the algorithm accordingly worsens by a small constant factor and we achieve a competitive ratio of $O(\locem\rewem^2\log \min(D,\lmax))$ as before. See Section~\ref{sec:extension_proofs} for further details.

    \begin{theorem}\label{t:multi-matching}
    Given an instance $(G, I, I', M)$ of the TW-TSP with predictions where $M$ is a many-to-one matching with errors as defined above, and satisfying $\twem\le\lmin/2$ and $\locem\leq \lmin/2$, there exists a polynomial-time online algorithm that takes the tuple $(G, I', \locem)$ as offline input and $I$ as online input, and constructs a walk $W\in \mathcal{W}(G)$ such that
    \[\mathbb{E}[\rew(W,I)]\ge \frac{1}{O(\locem\cdot \rewem^2\cdot \log \min(D,\lmax))}\cdot \opt(G,I)\]
    \end{theorem}

\end{itemize}

\section{Relating the Optima}\label{sec:relating_opt}
In this section we provide the proofs of Lemmas~\ref{l:service-costs} and~\ref{l:relating-opts} that relate the optima over the true and the predicted request sequences, using service times as a mechanism to capture the prediction errors. We begin by proving that a service time of $S$ can hurt the optimal by at most a factor of $2S-1$.

\newtheorem*{lemma:service-costs}{Lemma~\ref*{l:service-costs}}
\begin{lemma:service-costs}    
   For any instance $(G,I)$ of the TW-TSP with service times, and any integer $S\le\lmin$, we have \[\opt(G,I,S)\ge\frac{1}{2S-1}\cdot \opt(G,I,1).\]
\end{lemma:service-costs}

\begin{proof}

    Let $W\in\mathcal{W}(G)$ be the walk that achieves the optimum $\opt(G,I,1)$, and let the requests in $I$ that are covered by $W$ be denoted as $\sigma_i=(v_i, r_i, d_i, \pi_i)$ and ordered in the sequence in which they are covered by $W$. The lemma follows directly from the simple observation that if we don't service the $(S-1)$-requests \textit{prior} and \textit{after} some request $\sigma_i$, then we can save enough time to service $\sigma_i$ for $S$ time-steps within its time window.

    Formally, if $t_i\in [r_i, d_i-1]$ is the step at which $W$ begins servicing request $\sigma_i$, then by skipping the idle times on the $(S-1)$-previous and next requests we can remain idle on $v_i$ from step $t_i - (S-1)$ until step $t_i + S$ (since $W$ already remained idle on $v_i$ for $1$ step) while still being able to keep up with walk $W$. Since $S\leq L_{min}$, it is easy to verify that at least $S$ of these time-steps are going to fall in the time-window $[r_i,d_i]$.

    We now partition the requests $\sigma_i=(v_i, r_i, d_i, \pi_i)$ into $2S-1$ sub-sequences, each of which starts at some request $i\in [S]$, and covers the requests $\sigma_i, \sigma_{i+(2S-1)}, \sigma_{i+2(2S-1)},$ and so forth. Each such sequence can be covered with a walk, with idle times built in as above, so as to be feasible for the instance $(G, I, S)$. Clearly, one of these walks obtains a reward of at least $\opt(G, I, 1)/(2S-1)$, completing the proof.
\end{proof}

We show in Appendix~\ref{s:app-4proofs} that the above lemma obtains a tight gap between the optima at different service times.
\begin{lemma}\label{l:service-costs-lb}
    For any pair of integers $(L,S)$ such that $L \geq 2S-2\geq 1$, there exists a rooted instance $(G,I)$ of the TW-TSP with service costs such that $L_{min} = L$ and \[\opt(G,I,S) = \frac{1}{2S-1}\cdot\opt(G,I,1).\]
\end{lemma}

Next, we provide the proof of Lemma~\ref{l:relating-opts} that relates the optima between the predicted and true request sequences, by appropriately addressing all three possible types of prediction errors.

\newtheorem*{lemma:relating-opts}{Lemma~\ref*{l:relating-opts}}
\begin{lemma:relating-opts}
   Let $(G,I,I',M)$ be an instance of the TW-TSP with predictions, where $\locem$, $\rewem$, and $\twem$ denote the maximum location, reward, and time window errors of the instance respectively. Define $S:=4\locem+1$ and $S':=2\locem+1$. Then, if $\twem \leq \lmin/2$ and $\locem\leq (\lmin-1)/4$, we have \[\opt(G,I',S')\ge\frac{1}{3\rewem}\cdot \opt(G,I,S).\]
\end{lemma:relating-opts}

\begin{proof}
    Let $W$ be the walk that achieves the optimum $\opt(G,I,S)$, and let the requests in $I$ covered by $W$ be denoted as $\sigma_i=(v_i, r_i, d_i, \pi_i)$ and ordered in the sequence in which they are visited by $W$. Let $\sigma'_i=(v'_i, r'_i, d'_i, \pi'_i)$ denote the predicted request matched to $\sigma_i$, that is, $\sigma'_i=M(\sigma_i)$. Observe that the total reward of all requests $\{\sigma'_i\}$ corresponding to $\sigma_i\in\cov(W,I,S)$ is at least $\rew(W,I,S)/\rewem$.
    
    We will consider a walk $W'$ in $G$ defined as follows. The walk $W'$ follows $W$, visiting the requests $\sigma_i$ in sequence. As soon as $W$ starts servicing $\sigma_i$, $W'$ takes a detour to visit $\sigma'_i$; remains idle at $\sigma'_i$ for $S'$ time steps; returns back to $\sigma_i$; remains idle at $\sigma_i$ for $S-2\ell(v_i,v'_i)-S'\ge 0$ time steps; and then resumes the walk $W$. Observe that $W'$ is identical to $W$ outside of the detours it takes to visit the $\sigma'_i$'s. 


    Our goal is to feasibly capture all of the reward contained in the $\sigma'_i$s. The problem is that the walk $W'$ may miss some of this reward due to the mismatch in the time windows of the true and predicted requests. To this end, we will consider two variations of the walk $W'$. Let $K:=\lmin/2\ge \twem$. The walk $W'_1$ is identical to $W'$ except that it starts $K$ steps after $W'$ starts, and accordingly visits every location exactly $K$ steps after $W'$ visits it. The walk $W'_2$ is identical to $W'$ except that it starts $K$ steps {\em before} $W'$ starts,\footnote{To be precise, this walk starts at the location where $W'$ is at at step $K$.} and accordingly visits every location exactly $K$ steps before $W'$ visits it. 

    Now consider some $\sigma'_i$ corresponding to a request $\sigma_i$ covered by $W$ in the instance $(G,I,S)$. We claim that at least one of the walks $W'$, $W'_1$, and $W'_2$ covers $\sigma'_i$ in $(G,I',S')$. Let $t$ be the time at which $W'$ arrives at $v'_i$; recall that $W'$ remains at the node until at least $t+S'$. Note that $t\ge r_i$ and $t+S'\le d_i$ due to $\sigma_i\in\cov(W,I,S)$.
        
    First, suppose that $r'_i\le t$ and $d'_i\ge t+S'$, then $\sigma'$ is covered by $W'$ in $(G,I',S')$. Next suppose that $r'_i>t$. Then, $W'_1$ arrives at $v'_i$ at time $t+K\ge r_i+K\ge r_i+\twem\ge r'_i$. On the other hand, it remains at $v'_i$ until time $t+K+S'<r'_i+K+S'\le r'_i+\lmin\le d'_i$. Therefore, $\sigma'_i$ is covered by $W'_1$. Finally, suppose that $d'_i<t+S'$. Then, $W'_2$ arrives at $v'_i$ at time $t-K> d'_i-S'-K\ge d'_i-\lmin\ge r'_i$. On the other hand, it remains at $v'_i$ until time $t-K+S'\le d_i-K\le d_i-\twem\le d'_i$. Therefore, $\sigma'_i$ is covered by $W'_2$.

    We get that at least one of $W'$, $W'_1$, or $W'_2$ obtains at least a $1/3\rewem$ fraction of $\opt(G,I,S)$, where the factor of $\rewem$ is lost due to the mismatch in the predicted rewards. The lemma follows directly from this.
\end{proof}

\section{The offline approximation}\label{sec:offline}

In this section, we design an $O(\log \min(D,\lmax))$ deterministic and polynomial-time approximation algorithm for the TW-TSP with service times, providing the proof of Lemma~\ref{l:offline-alg}. 
Our proof relies on a series of reductions between different offline problems, applications of existing algorithms as well as the design of novel algorithmic components. We break up our argument into a series of lemmas which are proved in Appendix~\ref{s:app-5proofs}.


\begin{enumerate}
    \item First, we designing an $O(\log \min(D,\lmax))$-approximation algorithm for TW-TSP (without service times). Since the work of~\cite{CKP12} already provides a $O(\log L_{max})$ approximation for this setting, it suffices to prove the following:
    \begin{lemma}\label{l:large_window_approximation}
        Given an instance of the TW-TSP (without service times) with $L_{min}\geq 4D$, there exists a polynomial time algorithm that achieves an $O(1)$ approximation.
    \end{lemma}
    Our proof relies on the observation that when time-windows are sufficiently large compared to the diameter of the graph, the problem essentially reduces to an instance of the well-studied Orienteering problem, for which constant approximation algorithms are known. We comment that similar ideas have been used in~\cite{azar15}. Then, it is straight-forward to combine this algorithm together with the algorithm of~\cite{CKP12} to acquire an $O(\log \min(\lmax,D))$ approximation of TW-TSP.
    \begin{lemma}\label{l:tw_tsp_approx}
        There exists an $O(\log \min(D,\lmax))$ approximation algorithm for TW-TSP.
    \end{lemma}

    \item Next, we design a simple approximation-preserving reduction from TW-TSP with service times to TW-TSP (without service times). The main idea behind this reduction is to treat service times as edge lengths in an augmented graph whose diameter is roughly $D+S$. For instances with $S\leq D$, this increase becomes negligible and thus by combining our reduction with Lemma~\ref{l:tw_tsp_approx}, we immediately get the following:
    \begin{lemma}\label{l:tw_tsp_sc_approx}
        Given an instance $(G,I,S)$ of the TW-TSP with service times such that $S\le D$, there exists a polynomial time algorithm that achieves an $O(\log \min(D,\lmax))$ approximation.
    \end{lemma}

    \item Finally, we handle the case of large service times, specifically $S\geq D$. In that case, it turns out that we can reduce the instance to one over a uniform complete graph. Then, the TW-TSP-S essentially becomes equivalent to the well-studied \textit{Job Scheduling} problem, for which constant approximations are known.
    \begin{lemma}\label{l:large_service_cost}
        Given an instance $(G,I,S)$ of the TW-TSP with service costs such that $S\le D$, there exists a polynomial time algorithm that achieves an $O(1)$ approximation.
    \end{lemma}
\end{enumerate}
\noindent The proof of Lemma~\ref{l:offline-alg} follows immediately from Lemma~\ref{l:tw_tsp_sc_approx} and Lemma~\ref{l:large_service_cost}. We comment that any improvement in the best known approximation algorithm for TW-TSP will immediately imply an improvement for all the results that this work presents. Lemma~\ref{l:large_service_cost} essentially enables us to assume that in all instances of interest, $S\leq D$. Under this assumption, our reduction used in the proof of Lemma~\ref{l:tw_tsp_sc_approx} essentially states that TW-TSP-S becomes equivalent to TW-TSP in graphs of diameter $\mathcal{O}(D)$ and maximum window size $\mathcal{O}(\lmax)$. As an immediate corollary, given an offline TW-TSP algorithm that achieves a competitive ratio of $\alpha(D,\lmax)$, we immediately obtain an offline $O(\alpha(D,\lmax))$ approximation for TW-TSP-S, that can be used in order to substitute Lemma~\ref{l:offline-alg} in our analysis and improve the competitive ratio of Theorem~\ref{t:main}.


\section{The online algorithm}\label{sec:online}
In this section we present an online algorithm that takes as input a pre-computed walk over the predicted request sequence and solves an appropriate online matching problem in order to construct detours that capture true requests, while taking into account the possible errors in the predictions. The formal guarantee of our algorithm is given in Lemma~\ref{l:online-alg}, which we restate for the reader's convenience:
\newtheorem*{lemma:matching}{Lemma~\ref*{l:online-alg}}
\begin{lemma:matching}
    Given an instance $(G,I,I',M)$ of the TW-TSP with predictions satisfying $\twem\leq\lmin/2$ and $\locem\leq (\lmin-1)/4$; a walk $W'\in \mathcal{W}(G)$; and any integer $S'\geq 2\locem + 1$, there exists an online algorithm (Algorithm~\ref{alg:matching_alg}) that returns a walk $W\in\mathcal{W}(G)$ with expected reward
    \[\mathbb{E} \left[ \rew(W,I,1) \right] \ge \frac{1}{6\rewem}\cdot \rew(W',I',S').\]
\end{lemma:matching}

We begin by establishing some notation. Let $W'\in\mathcal{W}(G)$ be any walk that services some predicted requests in $I'$ with a service time of $S'$. We use $\sigma'_i=(v'_i, r'_i, d'_i, \pi'_i)$ to denote the predicted requests in $I'$ that are covered by $W'$, ordered in the sequence in which they are visited by $W'$. Likewise, we use $\sigma_i=(v_i, r_i, d_i, \pi_i)\in I$ to denote the true request matched to the prediction $\sigma'_i$, that is, $\sigma'_i=M(\sigma_i)$.

At a high level, our algorithm follows the walk $W'$, but when it reaches a predicted request $\sigma'_i$, it considers taking a detour to service a true request that is available at that point of time. To this end, we define the set of "reachable" true requests as follows.

\begin{definition}
Given a partial walk $W$ that is at request $\sigma'_i\in I'$ at time $t$, we define the set of reachable requests $R_i(W,t)$ to be the set of all $\sigma\in I$ such that:
\begin{enumerate}
    \item $r_\sigma\le t\le d_\sigma-\ell(v'_i, v_\sigma)-1$.
    \item $2\ell(v'_i, v_\sigma)+1\le S'$.
\end{enumerate}
\end{definition}

Our algorithm considers all of the reachable requests that have not been covered by the walk as yet, chooses the one with the highest reward, and takes a detour to visit and cover the request, before returning to $\sigma'_i$ and resuming the walk. In order to deal with time window errors, our algorithm starts the walk a little early, or on time, or a little late, as in the proof of Lemma~\ref{l:relating-opts}. The algorithm is described below formally.

\begin{algorithm}[h!]
\caption{Online algorithm for TSP-TW with predictions}
\label{alg:matching_alg}
{\bf Offline input:} Graph $G$, predicted requests $I'$, walk $W'\in \mathcal{W}(G)$, service times $S'$.\\
{\bf Online input:} True requests $I$.\\
{\bf Output:} Walk $W\in \mathcal{W}(G)$.
\begin{algorithmic}[1]
    \State Let $K=\lmin/2$. Select $\epsilon$ uniformly at random from $\{-1, 0, 1\}$.
    \State Define the set of covered requests $C=\emptyset$. 
    \For{$i\gets 1$ to $|\cov(W',I',S')|$}
        \State Let $t'_i$ denote the time at which $W'$ visits $\sigma'_i$. 
        \State Set $t_i\gets t'_i+\epsilon K$.
        \State Visit $v'_i$ at time $t_i$.
        \State Construct the set $R_i(W,t_i)$ of requests in $I$ reachable at time $t_i$. 
        \If {$R_i(W,t_i)\setminus C=\emptyset$}
            \State Do nothing.
        \Else
            \State Let $\hat{\sigma}$ be the highest reward request in $R_i(W,t_i)\setminus C$.
            \State Visit $v_{\hat{\sigma}}$; spend one unit of idle time at $v_{\hat{\sigma}}$; return to $v'_{i}$. \label{step:detour}
            \State Set $C\gets C\cup\{\hat{\sigma}\}$.
        \EndIf
    \EndFor
\end{algorithmic}
\end{algorithm}

We begin our analysis by noting that the walk $W$ constructed by the algorithm is always able to visit the vertices $v'_i$ corresponding to requests $\sigma'_i\in \cov(W',I',S')$ feasibly at the desired times $t_i$. This is because, by construction, the length of the detours that the walk $W$ takes in Step~\ref{step:detour} is always at most $S'$ -- the amount of idle time $W'$ spends at $v'_i$ -- by virtue of the fact that $\hat{\sigma}\in R_i(W,t_i)$ and therefore, $2\ell(v'_i,v_{\hat{\sigma}})+1\le S'$. Therefore, all of the requests $\hat{\sigma}$ visited in Step~\ref{step:detour} are indeed visited by the walk $W$. 

We now relate the total reward covered by $W$ to the reward contained in the true requests $\sigma_i$ corresponding to $\sigma'_i\in \cov(W',I',S')$. To do so, we first note that with constant probability each such request is reachable by $W$.


\begin{claim}\label{c:availability}
    For each $i$, $\sigma_i\in R_i(W,t_i)$ with probability at least $1/3$.
\end{claim}
\begin{proof}
    Recall that by definition we have $\sigma'_i=M(\sigma_i)$ and so, $2\ell(v'_i,v_i)+1\le 2\locem+1\le S'$. So the request $\sigma$ always satisfies the second requirement in the definition of the reachable set $R_i(W,t_i)$. Let us now consider the first requirement and recall that $t_i=t'_i+\epsilon K$ where $\epsilon\in\{-1,0,1\}$. We will now argue that $t_i \in [r_i, d_i - 1 -\ell(v'_i, v_i)]$ for at least one of the three choices of $\epsilon$. The claim then follows from the uniformly random choice of $\epsilon$.
    \begin{enumerate}
        \item If $t'_i \in [r_i, d_i - 1 -\ell(v'_i, v_i)]$, then the claim holds for $\epsilon=0$ and $t_i=t'_i$.
        \item Suppose that $t'_i<r_i$. Then, for $\epsilon=1$ we have that $t_i=t'_i+K \geq r'_i + \twem \geq r_i$, and also $t_i=t'_i + K < r_i + K < d_i - L_{min} + L_{min}/2$ and thus $t_i=t'_i + K \leq d_i -1 - \ell(v'_i, v_i)$ since $\ell(v'_i, v_i) \leq \locem \leq L_{min}/2$. Thus, in this case we have $t_i=t'_i+K \in [r_i, d_i - 1 -\ell(v'_i, v_i)]$ with the choice of $\epsilon=1$.
        \item Finally, suppose that $t'_i>d_i - 1 -\ell(v'_i, v_i)$. Then, for $\epsilon=-1$ we have that $t_i=t'_i - K \leq d'_i - S' - \twem \leq d_i - S'$ and thus $t'_i - K \leq d_i -1 - \ell(v'_i, v_i)$ since $S'\geq 2\locem + 1 \geq \ell(v'_i, v_i) + 1$. Also, $t_i=t'_i - K > d_i - 1 -\ell(v_i,v'_i) - L_{min}/2 > r_i + L_{min}/2 -\ell(v_i,v'_i)  - 1$ and thus $t'_i - K \geq r_i$, since $\ell(v_i,v'_i) \leq \locem \leq L_{min}/2$. Thus, in this case we have $t_i=t'_i-K \in [r_i, d_i - 1 -\ell(v'_i, v_i)]$ with the choice of $\epsilon=-1$.
    \end{enumerate}

\end{proof}

We are now ready to prove Lemma~\ref{l:online-alg} via a matching-type argument. To account for the reward covered by the walk $W$ constructed by the algorithm, we will employ a standard charging scheme. Every time the algorithm takes a detour to cover some true request $\hat{\sigma}$ from a predicted request $\sigma'_i$ in Step~\ref{step:detour}, we will credit half of the earned reward $\pi_{\hat{\sigma}}$ to $\hat{\sigma}$ itself, and half of the reward to the request $\sigma_i$. Formally, let $\cred(\sigma)$ denote the total credit received by $\sigma\in I$. Then during Step~\ref{step:detour} we will increment both $\cred(\hat{\sigma})$ and $\cred(\sigma_i)$ by $\pi_{\hat{\sigma}}/2$.

Now consider some $\sigma_i\in I$ corresponding to $\sigma'_i\in\cov(W',I',S')$. By Claim~\ref{c:availability}, this request is in $R_i(W,t_i)$ with probability at least $1/3$. If at time $t_i$, the request has already been covered by $W$, then we get $\cred(\sigma_i)\ge \pi_i/2$. Otherwise, we pick a $\hat{\sigma}\in R_i(W,t_i)$ with $\pi_{\hat{\sigma}}\ge \pi_i$, and therefore, once again we get $\cred(\sigma_i)\ge \pi_i/2$.

Putting everything together, we get
\begin{align*}
\operatorname{E}[\rew(W,I,S)]  = \operatorname{E}\left[\sum_{\sigma\in I} \cred(\sigma)\right] 
& \ge \sum_{i: \sigma'_i\in\cov(W',I',S')} \operatorname{E}\left[\frac{\pi_i}{2} \mathbbm{1}[\sigma_i\in R_i(W,t_i)]\right]\\
& \ge \frac{1}{3}\cdot \sum_{i: \sigma'_i\in\cov(W',I',S')} \frac{\pi_i}{2}\\
& \ge \frac{1}{6} \cdot \frac{1}{\rewem} \cdot \sum_{i: \sigma'_i\in\cov(W',I',S')} \pi'_i\\
& = \frac{1}{6\rewem}\cdot \rew(W',I',S')    
\end{align*}
This completes the proof of the lemma.

\section{Lower bounds}\label{sec:lower_bound}
In this section, we present lower bounds that complement our results. First, we will motivate the need for predictions in Section~\ref{sec:lb-online}. Then, in Section~\ref{s:loc-err-lb} we will show that the competitive ratio of TW-TSP with predictions must scale linearly with the error in locations. Finally, in Section~\ref{s:lb-service} we argue the need for non-zero service times in the definition of TW-TSP with predictions.

\subsection{Lower bounds for online TW-TSP without predictions}
\label{sec:lb-online}
We argue that Online TW-TSP does not admit any reasonable competitive ratio in the absence of predictions. In the case of deterministic algorithms where their entire behavior is predictable, simple instances with only $2$ vertices and appropriately small time-windows suffice to argue that no bounded guarantee for the approximation ratio is achievable.

\begin{lemma}\label{l:unbounded_deter}
The competitive ratio of any deterministic online algorithm for Online TW-TSP on instances with $\lmin \leq D$ is unbounded.
\end{lemma}
\begin{proof}
    Let $\deter$ by any deterministic algorithm and let $G$ be the line graph with just two vertices $v_1,v_2$ connected via an edge of length $D$. Since $\deter$ is deterministic, we can assume knowledge of its position at any step $t$ as soon as we have specified all requests with release time $\leq t$. We will now construct a request sequence $I$ that uses the information.

    For the first $D$ time-steps, we don't release any request. Then, at $t=D$, let $v_D\in \{v_1,v_2\}$ be the position that $\deter$'s walk is currently at and likewise let $v'_D$ be the other vertex of $G$.
    We construct the first request to be $\sigma=(v'_D,D,D+L,1)$ for any $L\leq D$. Clearly, $\deter$ cannot service this request as even for $L=D$ it arrives on $v'_D$ at deadline and cannot service it for one step. Then, we don't release any new request for the next $2D$ steps, and at $t=3D$ we repeat the same process, by requesting the vertex that $\deter$ doesn't currently occupy. Likewise, we repeat the same process at $t=5D$, $t=7D$ etc. Independently of the size of our request sequence, the total reward collected by $\deter$ is $0$.

    On the other hand, it is not hard to see that if our request sequence $I$ has $N$ requests in total, then $\opt(G,I,1)=N$. This is due to the fact that requests are spaced $2D$-steps from each other, and thus an optimal offline algorithm that had knowledge of the entire sequence in advance would always be able to arrive at each request on time, servicing it within its respective time-window.
\end{proof}

For randomized algorithms, a slight improvement can be achieved. In particular, the randomized algorithm that picks a vertex uniformly at random and then remains idle on it for the entire sequence achieves a competitive ratio of $1/n$. It turns out that this is actually the best possible ratio that a randomized algorithm can achieve on Online TW-TSP:

\newtheorem*{theorem:lb-unbounded}{Theorem~\ref*{l:unbounded_rand}}
\begin{theorem:lb-unbounded}
The competitive ratio of any randomized online algorithm for Online TW-TSP on instances with $\lmin \leq D$ is at most $1/n$.
\end{theorem:lb-unbounded}
\begin{proof}
    Let $G(V,E,\ell)$ be the uniform complete graph of $n=|V|$ vertices, where all edges have a length of $D$. Fix any integer $N$ and let vertices $v_1,v_2,\dotsc , v_N\in V$ be drawn independently and uniformly at random. Next, we fix any window length $L\le D$ and consider the (random) request sequence on these vertices $I=\{\sigma_i\}_{i=1}^N$ for $\sigma_i = (v_i, (2i-1)D, (2i-1)D+L, 1)$. From Yao's mininmax principle, 
    a lower bound on the (expected) competitive ratio of deterministic algorithms on this randomized instance will imply the same lower bound for randomized algorithms.

    Since the time-windows are spaced $2D$-away from each other, it is not hard to see that for any realization of $I$, $\opt(G,I,1)=N$. On the other hand, since $D\geq L$, we get that the only way to service a request is to be on its corresponding vertex on release time. Since the vertices are random, for any deterministic algorithm this happens with probability  precisely $1/n$, and thus the expected reward of any deterministic algorithm on this instance is $N/n$, proving the claim.
\end{proof}

\subsection{Tight dependence on location error}
\label{s:loc-err-lb}

In this section we show that a linear dependency on the location error is unavoidable for any randomized online algorithm for the TW-TSP with predictions, even assuming exponential computational power. In other words, we formally prove Theorem~\ref{t:lb}, which we re-state for the reader's convenience: 

\newtheorem*{theorem:lb}{Theorem~\ref*{t:lb}}
\begin{theorem:lb}
    For any $S>0$, there exists an instance $(G, I, I', M)$ of the TW-TSP with predictions satisfying $\twem=0$, $\rewem=1$, and $\locem=S$ such that the competitive ratio of any randomized online algorithm taking the tuple $(G, I', \locem)$ as offline input and $I$ as online input asymptotically approaches $1/(S+1)$.
\end{theorem:lb}

\begin{proof}
Fix any $S>0$ and let $K$, $C$ and $N$ be integer parameters that will be specified later. We construct a graph $G=\cup_{i=0}^{N-1}G_i$ that consists of $N$ copies $G_0,\dotsc , G_{N-1}$ of the complete graph on $C$ vertices with all edge lengths equal to $S$, arranged in a way so that each vertex in $G_i$ connects to each vertex in $G_{i+1}$ with an edge of length $KS$. A pictorial example for small values of $C$ and $N$ is shown in Figure~\ref{fig:lb}. 

\begin{figure}[h]
\centering
\includegraphics[width=0.8\textwidth]{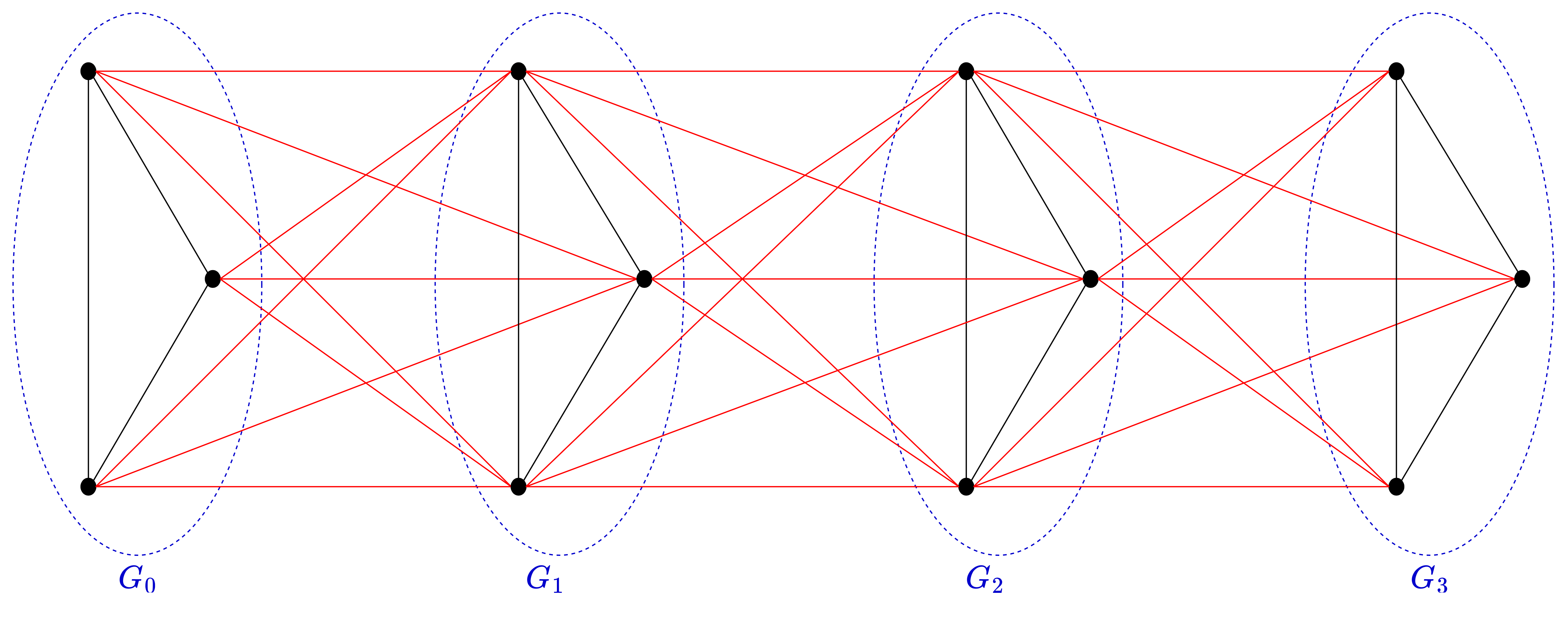}
\caption{An example of $G$ for $C=3$ and $N=4$. Black edges have a length of $S$ and red edges have a length of $KS$.}
\label{fig:lb}
\end{figure}

Next, we select independently and uniformly at random one vertex $v_i$ from each sub-graph $G_i$ and construct the (randomized) request sequence $I=\{\sigma_i\}_{i=0}^{N-1}$ where $\sigma_i = (v_i, r_i, r_i + KS, 1)$ for $r_i = i\cdot (KS+1)$. As for the predictions, we simply construct a second instance $I'$ in the same manner and offer it as an offline prediction for $I$. Observe that for all possible constructions of $I$ and $I'$ it holds that there exists a matching $M$ between them with $\twem=0$, $\rewem = 1$ and $\locem = S$, namely the matching that pairs together vertices from the same sub-graphs. This prediction provides no information to the algorithm other than the fact that the (true) instance $I$ was constructed via the above randomized approach. Also, notice that the location error $\locem$ can be arbitrarily small compared to the window length $L=KS$ by choosing appropriately large $K$.

It is easy to see that for all possible realizations of $I$ it holds that $\opt(G,I,1) = N$. Indeed, consider the walk that starts from the (random) vertex $v_0$, remains idle for $1$ step and then visits vertex $v_1$, remains idle for one step, visits $v_2$, etc. Such a walk would visit each vertex $v_i$ at step $t=i\cdot KS + i = r_i$ and thus would service all the requests in $I$, achieving a total reward of $N$. Next, we will show that the expected reward of any deterministic algorithm on the random sequence $I$ approaches $N/S$. Using Yao's minimax principle, this will immediately translate to a lower bound that approaches $1/S$ for randomized algorithms, completing the proof of the theorem.

Fix any deterministic algorithm for TW-TSP with predictions on instance $(G,I,I',M)$. Note that since the predictions $I'$ supply zero information, it suffices to analyze the algorithm as a deterministic online algorithm for Online TW-TSP on instance $I$. The key observation is that due to the fact that both the time windows of the requests and the edges that connect different sub-graphs have a length of $KS$, it is impossible for the algorithm to know on what vertex $v_i$ of subgraph $G_i$ the request is going to arrive before visiting some possibly different vertex of the same subgraph. In particular, if the algorithm is in subgraph $G_j$ for $j<i$ at time $r_i$, then it cannot reach vertex $v_i$ before the time window of request $i$ ends. Thus, in order for any deterministic algorithm to serve the request on sub-graph $G_i$, it first has to visit some vertex $v'_i$ of $G_i$. If it so happens that $v'_i=v_i$ then it can immediately service the request, otherwise it has to travel a distance of $S$ in order to reach $v_i$.

We partition the set of requests into two sets $N_+$ and $N_-$ based on whether the deterministic algorithm happens to arrive on the correct vertex of the sub-graph or not. All the requests in $N_+$ can be serviced without any extra delay, exactly as done by the optimal walk. On the other hand, servicing a request in $N_-$ requires the algorithm to spend an extra time of $S$ in order to transition to the right vertex. Since the time-windows have a length of $KS$, this can be done at most $K$ times before the algorithm runs out of slack to spare. When this happens, the algorithm would have to skip the next $S$ requests in order to recover enough slack to fix its next mistake. 

Formally, consider the last request in $N_-$ that the algorithm feasibly serves, call it $l$. Let $A$ be the number of requests in $N_-$ the algorithm serves through extra delay prior to $l$, and let $B$ be the number of requests the algorithm skips in $N_+$ or $N_-$ prior to $l$. Then in order for the algorithm to have reached $l$ before its time window ends, it must be the case that the total extra delay incurred by the algorithm, namely $AS-B$, is no more than $KS$. Rearranging we get:
\[A(S+1)-(A+B)\le KS, \quad\text{ or, }\quad A \le \frac{A+B}{S+1} + \frac{KS}{S+1} < \frac{N}{S+1} + K \]

\noindent
Thus, we get that the total expected reward gathered by the algorithm is at most
\[\mathbb{E}[|N_+|] + A + 1\le \mathbb{E}[|N_+|] + \frac{N}{S+1} + K + 1
\]
Finally, using the fact that $\mathbb{E}[|N_+|] = N/C$, we get that the competitive ratio of any deterministic algorithm on the random instance $I$ is at most
\[\frac{1}{C} + \frac{K+1}{N} + \frac{1}{S+1}\]
Choosing $N>>K$ and $C$ sufficiently large, we can make this competitive ratio asymptotically approach $1/(S+1)$ as desired.
\end{proof}

\subsection{Comparing the optimal with and without service times}
\label{s:lb-service}


As we saw in Lemma~\ref{l:service-costs}, the gap between $\opt(G,I,1)$ and $\opt(G,I,S)$ depends linearly on the service time $S$. In this section, we study the gap between $\opt(G,I,1)$ and $\opt(G,I,0)$, showing that there exists a much sharper separation between them.
\begin{theorem}\label{t:vs_no_service}
    For any integers $L$ and $D$, $L\le D$, there exists an offline instance $(G,I)$ of the TW-TSP where $D$ is the diameter of the network and every request has a time window length equal to $L$, such that $\opt(G,I,1)\le \frac{L}{D+1}\cdot \opt(G,I,0)$.
\end{theorem}

\begin{proof}
Consider the line graph $G$ with vertices $v_0$ through $v_{D}$ connected sequentially via edges of length $1$. Clearly, the diameter of $G$ is $D$. Next, consider the sequence of $(D+1)$-requests $I=\{\sigma_i\}_{i=0}^D$ where $\sigma_i = (v_i, i, L+i, 1)$. Observe that $\opt(G,I,0)=D+1$ as simply following the walk from $v_0$ to $v_D$ will cover all the requests if the service costs are $0$. 

On the other hand, it is not very hard to see that $\opt(G,I,1)=L$. First, observe that without loss we can assume that the optimal walk starts on $v_0$. If not, let $\sigma_i$ be the \textit{first} request served by the optimal. Since $\sigma_i$ cannot be served prior to step $r_i=i$, we could instead start at $v_0$ and take $i$ steps in order to reach $v_i$ and then follow the original walk, achieving precisely the same reward. 

Our argument is completed by the simple observation that any walk that starts from $v_0$ and has served $x$ requests with service cost $1$ \textit{cannot} visit vertex $v_j$ prior to step $j+x$. Thus, as soon as $x$ becomes $L$ all the future time-windows will be missed.    
\end{proof}

Theorem~\ref{t:vs_no_service} states that $\opt(G,I,0)$ is a much stronger benchmark than $\opt(G,I,1)$. However, it doesn't exclude the possibility of designing an algorithm that is competitive against this stronger benchmark $\opt(G,I,0)$. We will next show that no randomized online algorithm with predictions can obtain a better than linear competitive ratio against this benchmark. Intuitively, requiring the algorithm to spend non-zero time at each request also allows the algorithm to recover from possible mistakes due to the prediction errors by skipping requests that the optimum services with a delay of $1$. A similar approach is not possible in the $0$ service time setting.  


\newtheorem*{theorem:lb-for-0-sc}{Theorem~\ref*{t:lb_for_0_sc}}
\begin{theorem:lb-for-0-sc}
    For any $S>0$, there exists an instance $(G,I,I',M)$ of the TW-TSP with predictions and service times $0$, satisfying $\twem=0$, $\rewem=1$, and $\locem=S$, such that any randomized online algorithm taking the tuple $(G, I', \locem)$ as offline input and $I$ as online input achieves a reward no larger than $O(1/n)\cdot\opt(G,I,0)$. Here $n$ is the number of vertices in $G$.
\end{theorem:lb-for-0-sc}
\begin{proof}
We construct the same graph $G=\cup_{i=0}^{N-1}G_i$ as in Theorem~\ref{t:lb}, that consists of $N$ copies of the complete graph with edge lengths $S$, connected sequentially with edges of length $KS$ (see Figure~\ref{fig:lb}). As for the request sequence, we once again select independently and uniformly at random one vertex $v_i$ from each sub-graph $G_i$ and construct the (randomized) request sequence $I=\{\sigma_i\}_{i=0}^{N-1}$ where $\sigma_i = (v_i, iKS, (i+1)KS-1, 1)$; notice that release times are slightly different from Theorem~\ref{t:lb} to account for the absence of service costs. 

For the predictions, we simply construct a second instance $I'$ in the same manner and offer it as an offline prediction for $I$. By matching the requests on the same sub-graph together, we get $\twem = 0$, $\rewem = 1$ and $\locem = S$. As it clearly holds that $\opt(G,I,0) = N$ for any realization of $I$, to prove our theorem it suffices to argue that any deterministic algorithm  for TW-TSP with predictions on instance $(G,I,I',M)$ gets an expected reward of $\mathcal{O}(N/n)$ and apply \textit{Yao's minimax principle}. Furthermore, since the prediction $I'$ does not provide any information on $I$, it suffices to bound the reward of deterministic algorithms for Online TW-TSP on (online) instance $I$.

Fix any deterministic algorithm for Online TW-TSP on instance $I$. Observe that since edges between different sub-graphs have a length of $KS$ and time-windows have a length of $KS-1$, it is impossible for the algorithm to service some request $\sigma_i$ unless at $t=r_i$ it is already in some vertex of sub-graph $G_i$. For the same reason, it is not possible to service any request $\sigma_j$ after servicing some other request $\sigma_i$ with $i>j$. Finally, since all requests can be reached by their release time if the algorithm starts from a vertex in $G_0$, we can assume without loss that this is indeed the case.

We partition our set of $N$ requests into two sets $N_+$ and $N_-$ based on whether the deterministic algorithm \textit{happens} to arrive on the correct vertex of the sub-graph before the request's release time or not. From the random construction of our instance, we have that $\mathbb{E}[|N_+|] \leq N/C$. For requests in $N_-$, if the algorithm wishes to service them it has to take a detour of length $S$ in order to reach the correct vertex. Our proof relies on the fact that after servicing $K$ requests in $N_{-}$, the algorithm can no longer service any other request. To see this, let $v_j$ be the $K$-th request in $N_{-}$ that was serviced by the deterministic algorithm. As we can already established, any request $v_i$ with $i<j$ can no longer be serviced. On the other hand, since without loss the algorithm starts at a vertex of $G_0$, just reaching a vertex in $G_i$ while taking $K$ detours of length $S$ requires at least $iKS + K>d_i$ time-steps. Putting everything together, we get that the expected reward of the algorithm is at most $N/C+K = \mathcal{O}(N/n)$ by setting $K=\mathcal{O}(1)$ and $N=2K$, since $n=NC$.
\end{proof}


\section{Many-to-one matching}\label{sec:extension_proofs}
In this section we will consider the many-to-one matching model for prediction errors introduced in Section~\ref{sec:approach}. We will prove Theorem~\ref{t:multi-matching}:


\newtheorem*{theorem:multi-matching}{Theorem~\ref*{t:multi-matching}}
\begin{theorem:multi-matching}
    Given an instance $(G, I, I', M)$ of the TW-TSP with predictions where $M$ is a many-to-one matching with errors as defined above, and satisfying $\twem\le\lmin/2$ and $\locem\leq \lmin/2$, there exists a polynomial-time online algorithm that takes the tuple $(G, I', \locem)$ as offline input and $I$ as online input, and constructs a walk $W\in \mathcal{W}(G)$ such that
    \[\mathbb{E}[\rew(W,I)]\ge \frac{1}{O(\locem\cdot \rewem^2\cdot \log \min(D,\lmax))}\cdot \opt(G,I)\]
\end{theorem:multi-matching}

\paragraph{Proof of Theorem~\ref{t:multi-matching}.} Clearly, Lemmas~\ref{l:service-costs} and~\ref{l:offline-alg} still hold in the many-to-one matching setting, as they are independent of the matching definition and simply state properties of the (offline) TW-TSP problem with service times. Furthermore, Lemma~\ref{l:relating-opts} only considers the matching \textit{from} the true requests \textit{to} the predictions. Since each request is still matched to a unique prediction, our analysis goes through without any changes. Actually, observe that our definition of the location error $\locem$ is slightly stronger in that case, as it also includes the service time and the return time. As a consequence, we can weaken our assumptions to $\locem \leq \lmin/2$ and work with a service time of $S'=\locem$ rather than $2\locem+1$.

We combine all these components to pre-compute a walk $W'\in\mathcal{W}(G)$ such that
\[\rew(W',I',S')\ge \frac{1}{O(\locem\cdot\rewem\cdot\log \min(D,\lmax))}\cdot \opt(G,I)\]
for $S'=\locem$, assuming that $\locem \le \lmin/2$ and $\twem\le \lmin/2$. It remains to address the online component of our analysis, namely to argue how we can construct an online walk $W\in\mathcal{W}(G)$ such that
\[\E\left[\rew(W,I,1)\right]\ge \frac{1}{\mathcal{O}(\rewem)}\cdot \rew(W',I',S').\]

Recall that in the (standard) one-to-one matching setting, this is achieved via Algorithm~\ref{alg:matching_alg} whose formal guarantee is stated in Lemma~\ref{l:online-alg}. Algorithm~\ref{alg:matching_alg} constructs a matching between $I'$ and $I$ in an online fashion. For this many-to-one setting, we have to cover more than one true request per prediction. In order to achieve this, we solve instances of the \textit{Orienteering} problem, rooted at each predicted request we visit.

At a high level, our algorithm follows the walk $W'$, but when it reaches a predicted request $\sigma'_i$, it considers taking a detour of length $S'$ to gather as much reward from servicing true requests that are available at that point of time as possible. To that end, we update our definition of  "reachable" true requests in order to \textit{just} address the discrepancy in the predicted time-windows. 

\begin{definition}
Given a partial walk $W$ that is at request $\sigma'_i\in I'$ at time $t$, we define the set of reachable requests $R_i(W,t)$ to be the set of all $\sigma\in I$ such that $r_\sigma\le t\le d_\sigma-\ell(v'_i, v_\sigma)-1$.
\end{definition}

Notice that this definition is weaker than the one provided for the one-to-one setting, as the closeness constraint over the reachable requests will be imposed by restricting the length of the detour when approximating the corresponding orienteering problem. Our algorithm considers all of the reachable requests that have not been covered by the walk as of yet, and approximates the maximum reward cycle of length $S'$ that starts on the predicted request, covers a subset of the reachable true requests that have not been covered and then returns to the predicted vertex. Once again, in order to address possible errors in the predicted time-windows, our algorithm starts the walk a little early, or on time, or a little late. The algorithm is described below formally.

\begin{algorithm}[h!]
\caption{Online algorithm for TSP-TW with predictions in the one to many matching setting.}
\label{alg:matching_alg_one_to_many}
{\bf Offline input:} Graph $G$, predicted requests $I'$, walk $W'\in \mathcal{W}(G)$, service times $S'$.\\
{\bf Online input:} True requests $I$.\\
{\bf Output:} Walk $W\in \mathcal{W}(G)$.
\begin{algorithmic}[1]
    \State Let $K=\lmin/2$. Select $\epsilon$ uniformly at random from $\{-1, 0, 1\}$.
    \State Define the set of covered requests $C=\emptyset$. 
    \For{$i\gets 1$ to $|Cov(W',I',S')|$}
        \State Let $t'_i$ denote the time at which $W'$ visits $\sigma'_i$. 
        \State Set $t_i\gets t'_i+\epsilon K$.
        \State Visit $v'_i$ at time $t_i$.
        \State Construct the set $R_i(W,t_i)$ of requests in $I$ reachable at time $t_i$. 
        \State Use $\mathrm{ORIEN}$ to compute a cycle $W_i$ of length $|W_i|\le S'$ that is rooted at vertex $v'_i$ and collects as much reward as possible from the set $C_i\subseteq R_i(W,t_i)\setminus C$ with a service time of $1$.
        \State Let $C_i\subseteq R_i(W,t_i)\setminus C$ be the set of requests covered by $\mathrm{ORIEN}$.
        \State Take a detour by following $W_i$ and servicing requests in $C_i$; on return, remain idle for $S'-|W_i|\geq 0$ steps.
        \State Set $C\gets C\cup C_i$.
    \EndFor
\end{algorithmic}
\end{algorithm}

Our algorithm uses the $(2+\epsilon)$-approximation algorithm for rooted orienteering that is presented in~\cite{CKP12} as a sub-routine to decide which true requests will be served during the service times. We refer to this algorithm as $\mathrm{ORIEN}$. We comment that typically the end vertex is not specified in an instance of Orienteering, but nonetheless we can try all possible final vertices $v_j$ and shorten the walk length from $S'$ to $S'-\ell(v'_i,v_j)$, all in polynomial time. Furthermore, Orienteering is typically defined with service time $0$, but we can still address this by adding extra edges on the graph of length $1/2$ to new, \textit{request vertices}, and move the requests there with a service time of $0$. That way, service times can be expressed as edge lengths.

Since $\tau_M$ is defined as the maximum time-window error among all \textit{matched} pairs, and furthermore $\tau_M \leq K = \lmin/2$ and $\locem \geq 1 + 2\ell(v_{\sigma'},v_\sigma)$ for all $\sigma$ such that $M(\sigma)=\sigma'$, we can directly apply the proof of Claim~\ref{c:availability} to get the following:
\begin{claim}\label{c:availablity_one_to_many}
    For each $i$ such that $\sigma'_i\in \cov(W',I',S')$, and each $j$ such that $\sigma'_i = M(\sigma_j)$, it holds that $\sigma_j\in R_i(W,t_i)$ with probability at least $1/3$. 
\end{claim}

Next, let $N_i\subseteq I$ be the set of requests that are matched to prediction $\sigma'_i\in\cov(W',I',S')$. Let $U_i\subseteq N_i$ be the set of requests in $I$ that are mapped to prediction $\sigma'_i$ \textit{and} were not covered by the algorithm during the detours for predictions $\sigma'_1$ through $\sigma'_{i-1}$; i.e.  $U_i = \{\sigma_j: \sigma'_i = M(\sigma_j)\}\setminus \cup_{j<i}C_{j}$. By definition of the reward error we have that $\pi'_i \le \rewem \cdot \sum_{j:\sigma_j\in N_i}\pi_j$ and thus we get that
\[\rew(W',I',S') \le {\rewem}\cdot \sum_{i: \sigma'_i\in\cov(W',I',S')}\left(\sum_{j:\sigma_j\in N_i}\pi_j\right)\]

To account for the reward covered by the walk $W$ constructed by the algorithm, we employ the following charging scheme. Every time the algorithm takes a detour during the service time of some prediction $\sigma'_i$ and serves some request $\sigma\in C_i$ that is matched to prediction $\hat{\sigma}'$ (i.e. $\hat{\sigma}'=M(\sigma)$), we credit half of the reward to prediction $\sigma'_i$ and half of the reward to prediction $\hat{\sigma}'$. We use $\cred(\sigma')$ to denote the total credit received by prediction $\sigma'\in I'$. 

\begin{claim}
    For each $\sigma'_i\in \cov(W',I',S')$, it holds $\mathbb{E}\left[\cred(\sigma'_i)\right]\geq \Omega(1)\cdot \sum_{j:\sigma_j\in N_i}\pi_j$.
\end{claim}
\begin{proof}
    Fix any $\sigma'_i\in \cov(W',I',S')$. If $\sigma_j\in N_i\setminus U_i$, then $\sigma_j$ is covered by $W$ and since $\sigma'_i = M(\sigma_j)$, we get that $\cred(\sigma'_i)$ was incremented by $\pi_j/2$ at $\sigma_j$'s cover time. Thus,
    \begin{equation}\label{eq:one_to_many1}
       \cred(\sigma'_i) \ge \frac{1}{2}\cdot \sum_{j\in N_i\setminus U_i}\pi_j     
    \end{equation}

    We now analyze the requests $\sigma_j\in U_i$ that have not been covered prior to $\sigma'_i$'s detour. Observe that (i) each of these requests is reachable with probability at least $1/3$ due to Claim~\ref{c:availablity_one_to_many}, (ii) the set of reachable requests in $U_i$ can clearly by covered by one cycle of length $S'$ (due to $S'=\locem$ and $\locem$'s definition), and (iii) the use of $\mathrm{ORIEN}$ guarantees as a $(2+\epsilon)$-approximation of the maximum reward cycle, we get that the expected total reward collected during the detour of $\sigma'_i$ will be at least $\frac{1}{(2+\epsilon)}\cdot (1/3)\cdot \sum_{j:\sigma_j\in U_i}\pi_j$. Since half of this reward will be stored to $\sigma'_i$, we have
    \begin{equation}\label{eq:one_to_many2}
       \mathbb{E}\left[\cred(\sigma'_i)\right]\ge \frac{1}{6(2+\epsilon)}\cdot \sum_{j\in U_i}\pi_j     
    \end{equation}
    The claim is derived from equations~\eqref{eq:one_to_many1} and~\eqref{eq:one_to_many2}.
\end{proof}

\noindent The proof of Theorem~\ref{t:multi-matching} is completed via the above claim. In particular, we get that
\[\rew(W,I,S) = \sum_{\sigma'\in I'} \cred(\sigma')\ge \frac{\Omega(1)}{\rewem} \cdot \rew(W',I',S').\]




\newpage
\bibliographystyle{alpha}
\bibliography{bibliog}

\newpage
\appendix
\section*{Appendix}

\section{Omitted proofs from Section~\ref{sec:relating_opt}}
\label{s:app-4proofs}

\newtheorem*{lemma:service-costs-lb}{Lemma~\ref*{l:service-costs-lb}}
\begin{lemma:service-costs-lb}
    For any pair of integers $(L,S)$ such that $L \geq 2S-2\geq 1$, there exists a rooted instance $(G,I)$ of the TW-TSP with service costs such that $L_{min} = L$ and \[\opt(G,I,S) = \frac{1}{2S-1}\cdot\opt(G,I,1).\]
\end{lemma:service-costs-lb}
\begin{proof}
    Consider a line graph $G$ with $n=2S-1$ vertices, labeled $v_0$ through $v_{n-1}$. Each vertex $v_i$ is connected to vertices $v_{i-1}$ and $v_{i+1}$ through edges of length $\alpha>0$. Next, consider the request sequence
    $I = \{\sigma_i\}_{i=0}^{n-1}$ where $\sigma_0 = (v_0, 0, L, 1)$ and $\sigma_i = (v_i, d_i - L, d_i, 1)$ for all $i=1,\dotsc , n-1$, with $d_i = i\alpha + 2S -1$. Thus, we have one unit-reward request per vertex of $G$ that has a window length of $L$. The edge length $\alpha$ is irrelevant to the analysis and can be set to any value that guarantees non-negative release times, i.e. $\alpha = (L+1-2S)$.
    
    Observe that any walk that starts from $v_0$ at $t=0$ and services requests with a service cost of $S$ can only cover one request. Indeed, reaching vertex $v_j$ while having served one request cannot happen sooner than step $\alpha\cdot j + S > d_j - S$. Thus, $OPT(G,I,S)=1$ when the instance $(G,I)$ is rooted at vertex $v_0$.

    On the other hand, consider the walk that starts from $v_0$ and visits all vertices of $G$ sequentially while remaining idle for one time-step in all of them. Such a walk would remain idle on each vertex $v_j$ from $t_j=j + \alpha\cdot j$ until step $(t_j + 1)$. Since $n=2S-1$, it is always the case that $t_j \leq d_j - 1$. Furthermore, assuming that $L\geq 2S-2$, it is always the case that $t_j \geq r_j$. Thus, this walk successfully serves all requests in the unit service cost model, implying $OPT(G,I,1)=2S-1$.
\end{proof}

\bigskip
\section{Omitted proofs from Section~\ref{sec:offline}}
\label{s:app-5proofs}

\newtheorem*{lemma:large_window_approximation}{Lemma~\ref*{l:large_window_approximation}}
\begin{lemma:large_window_approximation}
        Given an instance of the TW-TSP (without service times) with $L_{min}\geq 4D$, there exists a polynomial time algorithm that achieves an $O(1)$ approximation.
\end{lemma:large_window_approximation}
\begin{proof}
Fix any instance $(G,I)$ of TW-TSP such that $L_{min}\geq 4D$. We begin by arguing that the large time-windows allow us to assume (within constant factors) that all release-times and deadlines are multiples of $K:=2D$. Indeed, let $I'$ be the request sequence derived from $I$ where each $\sigma_i = (v_i,r_i,d_i,\pi_i)\in I$ is substituted with a new request $\sigma'_i = (v_i, r'_i,d'_i,\pi_i)$ where each release time is rounded-up to the closest multiple of $K$ ($r'_i = K\cdot\lceil \frac{r_i}{K}\rceil$) and each deadline is rounded-down to the closest multiple of $K$ ($d'_i = K\cdot\lfloor \frac{d_i}{K}\rfloor$). We refer to such an instance as $K$-aligned. Notice that since $L_{min}\geq 4D = 2K$, no request is going to vanish from this transformation, meaning that $d'_i\geq r'_i + K$. Clearly, since $I'$ is identical to $I$ other than having smaller time-windows, we have that

\begin{equation}\label{eq:aligned1}
    \text{For any walk }W\in\mathcal{W}(G): \quad \rew(W,I,0)\geq \rew(W,I',0)
\end{equation}

Next, consider the walk $W^*\in\mathcal{W}(G)$ that achieves the optimal reward $\opt(G,I,0)$. Fix any request $\sigma_i\in I$ that is serviced by $W^*$ and let $t_i\in [r_i,d_i]$ be the step that servicing of the request starts. We distinguish between three different cases:
\begin{enumerate}
    \item If $t_i \in [r'_i, d'_i]$ then clearly walk $W^*$ also services request $\sigma'_i\in I'$.
    \item If $t_i < r'_i = K\cdot\lceil \frac{r_i}{K}\rceil$, then considered the walk $W^*_+$ that is identical to $W^*$ with the difference that it initially remains idle for $K$-steps. Such a walk would visit vertex $v_i$ at $t'_i = t_i + K$. It is not hard to see that from the definition of $r'_i$ and $d'_i$ as well as the fact that $d'_i\geq r'_i + K$, we easily obtain that $t'_i\in [r'_i,d'_i]$ and thus walk $W^*_+$ services request $\sigma'_i\in I'$.
    \item Likewise, if $t_i > d'_i =  K\cdot\lfloor \frac{d_i}{K}\rfloor$, we consider the walk $W^*_-$ that is identical to $W^*$ with the difference that it \textit{starts early} in the position that $W^*$ was at $t=K$. Such a walk would visit vertex $v_i$ at $t'_i = t_i - K$. Again, we can easily verify that $t'_i\in [r'_i,d'_i]$ and thus walk $W^*_-$ services request $\sigma'_i\in I'$. Observe that this argument wouldn't hold if $t_i < K$, as we skip the first $K$ time-steps of $W^*$, but from $t_i > d'_i \geq K$ we know that this event never happens.
\end{enumerate}
Thus, we have argued that for every request $\sigma_i\in I$ covered by the optimal walk $W^*$, there exists one of three walks that services request $\sigma'_i\in I'$. This implies that
\begin{equation}\label{eq:aligned2}
    \opt(G,I',0)\geq \frac{1}{3}\cdot \opt(G,I,0)
\end{equation}

From~\eqref{eq:aligned1} and~\eqref{eq:aligned2} we immediately obtain that any $\alpha$-approximation algorithm for the $K$-aligned instance $(G,I')$ can be trivially transformed into a $3\alpha$-approximation algorithm for instance $(G,I)$. Our next step is to obtain a constant approximation algorithm for TW-TSP on $K$-aligned instances via a non-standard reduction to the \textit{Orienteering} problem. We comment that a similar approach was taken in~\cite{azar15}. Our algorithm is stated in Algorithm~\ref{alg:aligned_approximation}, that is a $6$-approximation for $K$-aligned instances as we will shortly prove. Combining everything, we have designed a $18$-approximation algorithm for TW-TSP whenever $L_{min}\geq 4D$, as stated by the Lemma.

\begin{algorithm}[h!]
\caption{A constant approximation algorithm for TW-TSP in $2D$-aligned instances.}\label{alg:aligned_approximation}
\begin{algorithmic}[1]
    \smallskip
    \State Partition the time-horizon into phases of length $K$, with phase $i$ covering steps $[(i-1) K, iK]$.
    \smallskip
    \State Set $C\leftarrow \emptyset$.
    \smallskip
    \State At the beginning of phase $i$, do:
    \smallskip
    \State $\quad$ $Rem(i)\leftarrow \{\sigma_j\in I: r_j\leq (i-1)K \text{ and } d_j \geq i\cdot K\}\setminus C$.
    \smallskip 
    \State $\quad$ $P_i \leftarrow $ output of Orienteering sub-routine on graph $G$ with requests $Rem(i)$ and budget $K/2$.
    \smallskip 
    \State $\quad$ Travel from your current position $v_0$ to the root $v_i$ of path $P_i$ and execute the path $P_i$.
    \smallskip 
    \State $\quad$ Remain idle for the remaining $K-K/2 -\ell(v_0,v_i)\geq 0$ time-steps of the phase.
    \smallskip 
\end{algorithmic}
\end{algorithm}

Algorithm~\ref{alg:aligned_approximation} operates in phases of length $K$. Note that due to the alignment assumption, at the beginning of each phase we know the entire set of requests that will be available throughout this phase, as no request can arrive at a time-step that is not a multiple of $K$. We use $Rem(i)$ to refer to the set of these requests during phase $i$ that have not been already served by the algorithm. Clearly, this is precisely the maximal set of requests that the algorithm can hope to service in phase $i$. Furthermore, all deadlines will expire no sooner than the end of the phase, so the service order of the requests within the phase doesn't matter. 

Thus, our algorithm is faced with the decision of what is the maximum reward it can achieve by servicing requests in $Rem(i)$, using a path of length at most $K=2D$. This is precisely the definition of the well-studied \textit{Orienteering} problem, which is the special case of TW-TSP with $r_i=0$ and $d_i = D$ for all $i$. Due to the fact that the algorithm's position at the beginning of the phase might be inconvenient, we instead search for the best possible path of length $K/2$; the idea being that the other $K/2 = D$ time-steps can be used to travel at the root of the return walk, no matter where it lies on the graph. For the Orienteering sub-routine, we use the $2.5$-approximation deterministic algorithm of~\cite{CKP12}.

\paragraph{Analysis.} We are now ready to analyze the performance of Algorithm~\ref{alg:aligned_approximation}. We use $\alg(G,I)$ to denote the total reward achieved by our algorithm, and $\opt(G,I)$ to denote the optimal reward. Furthermore, we define $\adv$ to be the optimal offline algorithm on $(G,I)$ with the restriction that it does not gain any reward from servicing requests that are covered by $\alg$. Thus, if we use $\adv(G,I)$ to denote the total reward of this auxiliary algorithm, we have that $\opt(G,I) \leq \alg(G,I) + \adv(G,I)$.

Now fix some phase $i$ and let $A(i)\subseteq I$ to denote the set of requests covered by $\alg$ and $B(i)\subseteq I$ to denote the set of requests covered by $\adv$. The key observation is that all the requests in $B(i)$ need to (i) be released at the beginning of phase $i$, (ii) expire after the end of phase $i$ and (iii) not be covered by $\alg$ throughout its execution. Thus, we get $B(i)\subseteq Rem(i)$. This however implies that there exists some path of length $K$ (the one that $\adv$ follows in phase $i$) that gets a total reward of $\pi(B(i))$ (by slightly abusing notation) from requests in $Rem(i)$ and thus there also exists some path of length $K/2$ with total reward $\pi(B(i))/2$ from requests exclusively in $Rem(i)$.

Our argument is completed by the fact that the requests in $A(i)$ are computed via a $2.5$-approximation algorithm for the Orienteering problem. From the approximation guarantee of the used sub-routine, we know that each path of length $K/2$ that covers requests in $Rem(i)$ has a total reward of at most $2.5\cdot \pi(A(i))$. Thus, we get that $\pi(B(i)) \leq 5\cdot \pi(A(i))$ and thus $\alg(G,I) \geq 0.2 \cdot \adv(G,I)$. Combining this with $\opt(G,I) \leq \alg(G,I) + \adv(G,I)$, we get that Algorithm~\ref{alg:aligned_approximation} is a $6$-approximation for $K$-aligned instances with $K=2D$, completing the proof.
\end{proof}

\bigskip
\newtheorem*{lemma:tw_tsp_approx}{Lemma~\ref*{l:tw_tsp_approx}}
\begin{lemma:tw_tsp_approx}
    There exists an $O(\log \min(D,\lmax))$ approximation algorithm for TW-TSP.
\end{lemma:tw_tsp_approx}
\begin{proof}
Fix any instance $(G,I)$ of TW-TSP. We partition the request sequence $I$ into two sub-sequences $I_1$ and $I_2$, such that $I_1$ includes all requests $\sigma_j\in I$ with time-window $d_j-r_j\geq 4D$ and $I_2$ includes all the other requests. Since all requests in $I_1$ have a window-length of at least $4D$, we can use the algorithm of Lemma~\ref{l:large_window_approximation} to get a walk $W_1$ such that
\[\rew(W_1,I_1)\geq \frac{1}{\mathcal{O}(1)}\cdot\opt(G,I_1)\]

Furthermore, since all requests in $I_2$ have a window-length of at most $4D$, we can use the $\mathcal{O}(\log L_{max})$-approximation algorithm of~\cite{CKP12} to get a walk $W_2$ such that
\[\rew(W_2,I_2)\geq \frac{1}{\mathcal{O}(\log D)}\cdot\opt(G,I_2)\]

Clearly, since $I_1$ and $I_2$ form a partition of $I$ it holds that $\opt(G,I_1) + \opt(G,I_2) \geq \opt(G,I)$. Furthermore, for each $j\in\{1,2\}$ we have that $I_j\subseteq I$ and thus $\rew(W_j, I)\geq \rew(W_j,I_j)$. Combining these inequalities we immediately get that if we output the walk from $\{W_1,W_2\}$ that achieves the highest reward, we are guaranteed a $\mathcal{O}(\log D)$-approximation of $\opt(G,I)$. Combined once again with the $\mathcal{O}(\log \lmax)$ approximation of~\cite{CKP12}, we get the lemma.
\end{proof}

\bigskip
\newtheorem*{lemma:tw_tsp_sc_approx}{Lemma~\ref*{l:tw_tsp_sc_approx}}
\begin{lemma:tw_tsp_sc_approx}
        Given an instance $(G,I,S)$ of the TW-TSP with service times such that $S\le D$, there exists a polynomial time algorithm that achieves an $O(\log \min(\lmax,D))$ approximation.
\end{lemma:tw_tsp_sc_approx}
\begin{proof}
Fix any instance $(G,I)$ of TW-TSP-S with service cost $S$. We construct a new graph $G_0$ from $G$ in the following manner:
\begin{itemize}
    \item For each request $\sigma_i=(v_i,r_i,d_i,\pi_i)\in I$, we add a new vertex $v_{\sigma_i}$ on $G$. Thus, we have that $V(G_0) = V(G)\cup \{v_{\sigma_i} | \sigma_i \in I\}$.

    \item For each added vertex $v_{\sigma_i}$, we add a new edge that connects it to vertex $v_i$. Thus, we have that $E(G_0)=E(G)\cup\{(v_i,v_{\sigma_i})|\sigma_i\in I\}$.

    \item We double the edge length of each existing edge and set the edge length of each newly added edge to $S$, i.e. $\ell'_e = 2\ell_e$ for all $e\in E$ and $\ell'(v_i,v_{\sigma_i}) = S$ for all $\sigma_i\in I$. Notice that graph $G_0$ has diameter $D_0 \leq 2(D+S)$. 
\end{itemize}

Next, we construct a request sequence $I_0$ from $I$, by substituting each $\sigma_i = (v_i,r_i,d_i,\pi_i)\in I$ with a new request $\sigma'_i = (v_{\sigma_i}, 2r_i+S, 2d_i - S, \pi_i)\in I_0$. Notice that $I_0$ has maximum window length $\lmax^0\leq 2\lmax$. As we will see, the instance $(G,I)$ for TW-TSP-S with service cost $S$ and the instance $(G_0, I_0)$ for TW-TSP without service cost are essentially equivalent. Once this is established, we can use the approximation algorithm of Lemma~\ref{l:tw_tsp_approx} for TW-TSP to approximate instance $(G_0,I_0)$ up to a factor of $\mathcal{O}(\log \min(\lmax^0, D_0)) = \mathcal{O}(\log\min (\lmax, D+S))$ which will immediately translate into a $\mathcal{O}(\log \min(\lmax, D+S))$-approximation algorithm for TW-TSP-SC with service cost $S$. Finally, we obtain the proof of Lemma~\ref{l:tw_tsp_sc_approx} from the assumption $S\leq D$.

For the rest of the proof, we establish that any solution of $(G,I,S)$ can be transformed into an equal reward solution of $(G_0,I_0,0)$ and vice versa. The proof is almost immediate from our definition of the graph $G_0$ and the sequence $I_0$. Let $W_0\in\mathcal{W}(G_0)$ be any walk that services some requests $\{\sigma'_i\}\in I_0$ with zero service cost. If $W_0$ services some request $\sigma'_i\in I_0$, then we know that it visits vertex $v_{\sigma_i}$ at some point in $[2r_i+S,2d_i-S]$. Since there is only one way to reach vertex $v_{\sigma_i}$, we can deduce that walk $W_0$ visits vertex $v_i$ no sooner than $t=2r_i$ and it returns back to it after $2S$ time-steps, no later than $t=2d_i$.  

Let $W\in\mathcal{W}(G)$ be the walk that visits vertices in the same order as $W_0$, but instead of taking detours of the form $(v_i, v_{\sigma_i}, v_i)$, it simply remains idle on $v_i$. Also, due to the doubling of edges, all times are cut in half. Since $W_0$ visits each  vertex $v_{i}$ after $t=2r_i$ and takes a detour of length $2S$, returning by $t=2d_i$, we immediately get that this walk $W$ will service request $\sigma_i\in I$ with a service cost of $S$. Thus, $W$ will collect as much reward in instance $(G,I,S)$ as $W_0$ collects in instance $(G_0,I_0,0)$.

The proof of the other direction (starting from a walk in $G$ with service cost $S$ and achieving a same-reward walk in $G_0$ with service cost $0$) is completely analogous, and it is thus omitted.
\end{proof}

\bigskip
\newtheorem*{lemma:large_service_cost}{Lemma~\ref*{l:large_service_cost}}
\begin{lemma:large_service_cost}
    Given an instance $(G,I,S)$ of the TW-TSP with service costs such that $S\le D$, there exists a polynomial time algorithm that achieves an $O(1)$ approximation.
\end{lemma:large_service_cost}

\begin{proof}
Fix any instance $(G,I)$ of the TW-TSP-SC with service cost $S\geq D$. We begin by reducing this instance to the case of uniform graphs. Let $G_U$ be a new graph that is constructed by $G$ via connecting all vertices and increasing all edge length to the diameter of $G$; i.e., $G_U$ is the complete graph with $V(G_U)=V(G)$ and edge lengths $D$. Observe that since we only increase distances, the instance $(G_U,I)$ is at least as hard as the instance $(G,I)$, implying that for any walk $W_U\in\mathcal{W}(G_U)$, there exists some other walk $W\in\mathcal{W}(G)$ such that 
\begin{equation}\label{eq:rew_in_uniform}
\rew(W,G,I,S)\geq \rew(W_U,G_U, I, S),   
\end{equation}
namely the walk $W$ that follows $W_U$ and adds suitable idle times whenever distances in $G_U$ are larger than in $G$.

Next, we will show that the optimals of these instances can only differ by a constant factor. To see this, let $W^*\in\mathcal{W}(G)$ be the walk that achieves the optimal reward $\opt(G,I,S)$ and let $\{\sigma_i\}\subseteq I$ be the sequence of requests that $W^*$ covers, ordered in the sequence at which they are visited by $W^*$. We will now construct two walks $W_o,W_e\in\mathcal{W}(G_U)$ for instance $(G_U,I)$, respectively servicing all odd requests $\sigma_{2i+1}$ and all even requests $\sigma_{2i}$. 

Since the distance between any two vertices is increased by at most $D$ from $G$ to $G_U$ and the service cost satisfies $S\geq D$, we get that by skipping the service time of every other request in the sequence $\{\sigma_i\}$ and following $W^*$, we gain enough time to make up for the increased distances. Formally, $W_e$ is derived from $W^*$ by skipping an idle time of $S'-\ell(v_{2i+1},v_{2i+2}) \geq 0$ on each odd request $\sigma_{2i+1}$ and likewise, $W_o$ is derived from $W^*$ by skipping an idle time of $S'-\ell(v_{2i},v_{2i+1}) \geq 0$ on each even request $\sigma_{2i}$. We comment that for $W_o$, in order to service the first request (since skipping before it is not possible), we can simply modify our walk to start from vertex $v_1$ and add appropriate idle time. Thus, we acquire
\begin{equation}\label{eq:opt_in_uniform}
\opt(G_U,I,S)\geq \frac{1}{2}\cdot \opt(G,I,S)
\end{equation}

From~\eqref{eq:rew_in_uniform} and~\eqref{eq:opt_in_uniform}, we immediately get that any $\alpha$-approximation algorithm for TW-TSP-SC on uniform graphs directly implies a $2\alpha$-approximation algorithm for TW-TSP-SC in general graphs with $D\leq S$. Our proof is completed by the fact that TW-TSP-SC on uniform graphs is essentially equivalent to the job scheduling problem studied in~\cite{BGNS01} and~\cite{BBFNS01}, for which constant approximation algorithms are known. 

To see this, observe that since all distances are set to $D$, we can remove the travelling component from the formulation of the problem. Specifically, instead of requiring each request $\sigma_i$ to be serviced in some interval $[t,t+S]$ for $t\in [r_i, d_i-S]$, we can require it to be serviced in some interval $[t,t+S+D]$ for $t\in [r_i, d_i - S]$ that also incorporates the travelling time of $D$ that is required to get to the next request. Thus, we can view each request $\sigma_i$ as a job $J_i$ with release time $r_i$, deadline $d_i+D$, processing time $S+D$ and reward $\pi_i$. Then, our objective is to schedule these jobs in a single machine as to not have any overlaps and maximize the total reward of the jobs processed within their time-windows. This is precisely the formulation of the problem studied in~\cite{BGNS01}, where a $3$-approximation deterministic algorithm is given, later improved to a $(2+\epsilon)$-approximation deterministic algorithm in~\cite{BBFNS01}. In any case, using any of these algorithms, we can get a $\mathcal{O}(1)$-approximation for TW-TSP-SC on uniform graphs, completing our proof of the Lemma.
\end{proof}

\bigskip
\section{Omitted proofs from Section~\ref{sec:approach}}
\label{s:app-ext-proofs}

\newtheorem*{c:random-rewards}{Corollary~\ref*{c:random-rewards}}
\begin{c:random-rewards}
    Given an instance $(G, I, I', M)$ of the TW-TSP with predictions where requests have randomly drawn rewards, and predictions errors satisfy that $\twem\le\lmin/2$ and also $\locem\le (\lmin-1)/4$, there exists a polynomial-time online algorithm that takes the tuple $(G, I', \locem)$ as offline input and $I$ as online input, and constructs a walk $W\in \mathcal{W}(G)$ such that
    \[\mathbb{E}[\rew(W,I)]\ge \frac{1}{O(\locem\cdot\rewem^2\cdot\log \min(D,\lmax))}\cdot \opt(G,I)\]
\end{c:random-rewards}

\begin{proof}
Recall that by definition, in the random reward setting we have
\begin{equation}\label{eq:random1}
   \opt(G,I) = \argmax_{W\in\mathcal{W}(G)}\left( \mathbb{E} \left[\sum_{\sigma\in \cov(W,I)}\pi_\sigma\right]\right) = \argmax_{W\in\mathcal{W}(G)} \left(\sum_{\sigma\in \cov(W,I)}\mathbb{E} \left[ \pi_\sigma\right]\right) =  \opt(G,I_{det}) 
\end{equation}
where $I_{det}$ is the deterministic request sequence where each request $\sigma\in I$, rather than having a random reward $\pi_\sigma$, has a deterministic reward $\mathbb{E}[\pi_\sigma]$.  We comment that when exchanging the sum with the expectation, we made use of the fact that only the rewards in $I$ are random, and thus for any walk $W\in\mathcal{W}(G)$, the set $\cov(W,I)$ is deterministic.

Recall that in this random setting, $\rewem$ captures the maximum discrepancy between the \textit{expected} rewards $\mathbb{E}[\pi_\sigma]$ and the predicted rewards $\pi_{\sigma'}$ and is thus equal to the reward error between $I'$ and $I_{det}$. Furthermore, the location and time-window errors have nothing to do with the rewards and are thus identical between $(I,I')$ and $(I_{det}, I')$. Thus, we can directly apply Theorem~\ref{t:main}, to obtain an online walk $W\in\mathcal{W}(G)$ such that
\begin{equation}\label{eq:random2}
   \mathbb{E}\left[rew(W,I_{det}) \right]\ge \frac{1}{O(\locem\cdot\rewem^2\cdot\log \min(D,\lmax))}\cdot \opt(G,I_{det}),
\end{equation}
where the expectation is taken over the algorithm's random choices.

Finally, for any walk $W\in\mathcal{W}(G)$, we have that
\begin{equation}\label{eq:random3}
   \rew(W,I_{det}) = \sum_{\sigma \in \cov(W,I_{det})}\mathbb{E}\left[\pi_\sigma\right] 
    = \mathbb{E}\left[\sum_{\sigma \in \cov(W,I)}\pi_\sigma\right] = 
    \mathbb{E}\left[\rew(W,I)\right]
\end{equation}
since $\cov(W,I_{det}) = \cov(W,I)$. In this case, the expectation is taken over only the randomness of $I$. The proof follows from equations~\eqref{eq:random1} through~\eqref{eq:random3}.
\end{proof}

\bigskip
\newtheorem*{cor:incomplete_matching}{Corollary~\ref*{cor:incomplete_matching}}
\begin{cor:incomplete_matching}
    Given an instance $(G,I,I')$ of the TW-TSP with predictions, let $M$ be any (incomplete) matching between $I$ and $I'$, and let the error parameters $\locem, \rewem, \twem, \Delta_1^M,$ and $\Delta_2^M$ be defined as above. Then, there exists an online algorithm that takes $(G, I', \locem)$ as offline input and $I$ as online input, and returns a walk $W\in\mathcal{W}(G)$ such that
     \[\mathbb{E}\left[\rew(W,I)\right] \geq \Omega\left(\frac{1}{\locem\cdot \rewem^2\cdot \log \min(D,\lmax)}\right)\cdot \left(\opt(G,I) - \Delta_1^M\right) - \frac{\Delta_2^M}{\rewem}.\]
\end{cor:incomplete_matching}
\begin{proof}
    Let $I_1\subseteq I$ be the request sequence that consists of only true requests in $I$ that are matched to some predicted request in $I'$. Observe that the Lemma~\ref{l:service-costs} and Lemma~\ref{l:offline-alg} have nothing to do with the matching $M$ between $I$ and $I'$. Furthermore, Lemma~\ref{l:relating-opts} only requires that all requests in $I$ have a matched prediction in $I'$, so it can immediately be applied to $I_1$. Combining everything, we can pre-compute a walk $W'\in\mathcal{W}(G)$ such that
    \[\rew(W',I',S') = \Omega\left(\frac{1}{\locem\cdot \rewem^2\cdot \log \min(D,\lmax)}\right)\cdot OPT(G,I_1)\]
    for $S'\geq 2\locem + 1$. Next, let $I_2\subseteq I'$ be the predicted request sequence that consists of only predictions that are matched to some true request. Since all predictions in $I_2$ are matched to some true request in $I$, we can use Lemma~\ref{l:online-alg} with walk $W'$ as the input to construct an online walk $W\in\mathcal{W}(G)$ such that
    \[\mathbb{E}\left[\rew(W,I)\right] \geq \frac{1}{6\rewem}\cdot \rew(W', I_2, S')\]
    The corollary follows from the definition of the unmatched request rewards $\Delta_1^M$ and $\Delta_2^M$ which implies that $\opt(G,I_1,S)\geq \opt(G,I,S) - \Delta_1^M$ and $\rew(W', I_2, S')\geq \rew(W', I', S') - \Delta_2^M$.
\end{proof}


\end{document}